\documentclass[12pt]{article}
\usepackage[margin=1in]{geometry} 
\usepackage{amsmath,amsthm,amssymb,todonotes}
\usepackage{multirow}
\usepackage{times}
\usepackage{listings}
\usepackage{graphicx}
\usepackage{enumerate}
\usepackage[colorlinks,
linkcolor=blue,
anchorcolor=blue,
citecolor=blue
]{hyperref}
\usepackage{dsfont}
\usepackage{float}
\usepackage[title]{appendix}
\usepackage{booktabs}
\usepackage{caption}
\usepackage{tikz}
\usetikzlibrary{decorations.pathreplacing}

\usepackage{chngcntr}

\usepackage{natbib}

\usepackage{setspace}
\newcommand{\N}{\mathbb{N}}

\newcommand{\R}{\mathbb{R}}

\makeatletter

\newcommand{\Rmnum}[1]{\expandafter\@slowromancap\romannumeral #1@}
\makeatother
\allowdisplaybreaks

\newtheorem{theorem}{Theorem}[section]
\newtheorem{proposition}[theorem]{Proposition}
\newtheorem{lemma}[theorem]{Lemma}
\newtheorem{definition}[theorem]{Definition}
\newtheorem{corollary}[theorem]{Corollary}

\newtheorem{assumption}{Assumption}

\newtheorem*{claim}{Claim}

\numberwithin{equation}{section}

\title{Robust Data-Driven Decisions Under Model Uncertainty\footnote{This is the accepted version at EC$^{\prime}$22. I am deeply indebted to my advisors, Peter Klibanoff and Marciano Siniscalchi for constant guidance and support throughout the completion of this paper. I benefited greatly from discussions with Eran Shmaya and Isaias N. Chaves. I thank Federico Bugni, Ivan Canay, Shaowei Ke, Yijun Liu, Edwin Munoz Rodriguez, and Chen Zhao for insightful discussions. All errors are my own.}}

\author{Xiaoyu Cheng\footnote{
		Department of Managerial Economics and Decision Sciences, Kellogg School of Management, Northwestern University, Evanston, IL, USA. E-mail: xiaoyu.cheng@kellogg.northwestern.edu}
}
\begin{document}
	\maketitle
	
	\begin{abstract}
		When sample data are governed by an unknown sequence of independent but possibly non-identical distributions, the data-generating process (DGP) in general cannot be perfectly identified from the data. For making decisions facing such uncertainty, this paper presents a novel approach by studying how the data can best be used to \textit{robustly improve} decisions. That is, no matter which DGP governs the uncertainty, one can make a better decision than without using the data. I show that common inference methods, e.g., maximum likelihood and Bayesian updating cannot achieve this goal. To address, I develop new updating rules that lead to robustly better decisions either asymptotically almost surely or in finite sample with a pre-specified probability. Especially, they are easy to implement as are given by simple extensions of the standard statistical procedures in the case where the possible DGPs are all independent and identically distributed. Finally, I show that the new updating rules also lead to more intuitive conclusions in existing economic models such as asset pricing under ambiguity.\\
		
		
	\end{abstract}
	
	\section{Introduction}
	In data-driven decisions, the extrapolation from sample data is often based on observable similarities within a population. However in many settings, one also needs to take into account possible unobserved heterogeneity across individuals in the population. For instance, online platforms such as Netflix and Google rely on users' feedback to make content recommendations. While they can categorize the users according to their online activities and profiles, the users' preferences are usually also determined by their various offline activities that cannot be observed by those online platforms. Physicians use conclusions from clinical trials to guide their medication decisions. But effectiveness of the drugs can also depend on underlying health or genetic conditions that may not be documented in those trials. In these setting, because the heterogeneity is unobservable, the decision-maker (DM) can be uncertain about not only how the individuals might vary but also how different individuals are going to be sampled from the population. As a result, she might worry about the possibility that the sample data are more from one type of individuals, whereas future draws that determine the payoff of decisions may be more of a different type.\footnote{For example, a well-documented problem in clinical trials is the fact that the population of clinical trial participants is often different from the actual population of patients, because the former population responds to either medical or financial incentives \citep{Manski2013}. When it comes to decisions, an article in \textit{Nature} \citep{Schork2015} finds that many drugs that are prescribed based on traditional clinical trials help only between 1 in 25 and 1 in 4 of the people who take them and sometimes can even be harmful.} More precisely, she might believe that the sample data and future uncertainty can potentially be governed by different probability distributions. 
	
	To capture this concern, in this paper, I consider a decision environment where the underlying data-generating process (DGP) is a sequence of independent but possibly non-identical distributions.\footnote{The independent here is more precisely conditional independence. Most of the results in this paper do not hinge on this assumption and can be suitably generalized to allow for correlations.} Specifically, the DM observes sample data given by realizations of marginal distributions of a DGP and then makes a decision whose payoff depends only on future realizations of the same DGP. As the DGP can be a sequence of non-identical distributions, which captures the possible heterogeneity across individuals, the sample data and future uncertainty indeed can be governed by different probability distributions. In addition, because the heterogeneity is unobserved, the DM will also face \textit{model uncertainty} in the sense that she is uncertain about which probabilistic model\footnote{in the terminology of, e.g., \cite{Hansen2014} and \cite{Marinacci2015}.}, or DGP here, governs the underlying uncertainty. Suppose that the DM initially only knows there is a set of possible DGPs and cannot for any probabilistic assessment over them. In other words, the model uncertainty here is represented by the ambiguity over a set of DGPs. 
	
	Notice that, with non-identical distributions, the true DGP in general cannot be perfectly identified from the sample data. Nonetheless, learning from the data can potentially help the DM to refine her initial knowledge about the underlying uncertainty, which in turn may lead to a better decision under the true DGP. Motivated by this intuition, I investigate the following novel question: How can the DM use sample data to robustly improve her decisions? That is, no matter which possible DGP governs the uncertainty, she can do better than without using the data. To answer this question, I make two main contributions in this paper. First, I characterize a crucial property for inference rules that is necessary and sufficient to guarantee the decisions to be robustly better. Second, I develop novel and tractable inference (updating) rules that satisfy this crucial property and thus lead to robustly better decisions. Importantly, I also show that common inference rules, such as maximum likelihood and Bayesian updating, cannot achieve this goal. Instead, they can almost surely lead to strictly worse decisions than simply ignoring the data.\footnote{See Section \ref{example} for an illustrative example.} 
	
	To be more precise, as the DM faces ambiguity over a set of DGPs, motivated by making a robust decision, I suppose that the DM applies the maxmin expected-utility (MEU) criterion \citep*{GILBOA1989141} to cope with the ambiguity. In other words, she does not know which DGP is true and thus makes an optimal decision considering the worst possible one. In the presence of sample data, the DM can choose to ignore the data and make a \textit{data-free decision} based only on her initial knowledge, i.e., optimal under the worst DGP among all those she initially contemplates. On the other hand, she can also use the data to update the initial set of DGPs to a potentially different \textit{updated set} of DGPs. The \textit{data-driven decision} is then formally defined as the optimal decision under the worst-case DGP among only those in the updated set. 
			
	Given these two types of decisions, I study updating rules in terms of how to guarantee the data-driven decisions to be better than the data-free decisions according to \textit{objective payoff}, i.e., the expected utility under the \textit{true} DGP that governs the future uncertainty. In other words, while the DM makes decisions considering the worst-case, the quality of her decisions will be evaluated against the ground truth. When an updating rule can guarantee improvement for all possible DGPs in the initial set, the data-driven decisions are indeed robustly better than the data-free decisions. 
	
	In this paper, I formalize two achievable notions that the data-driven decisions robustly improve upon the data-free decisions across decision problems. Specifically, in the first notion, I restrict attention to \textit{basic decision problems}, i.e., binary choices between an uncertain and a constant prospect. For instance, the decisions of whether or not to recommend a piece of content, or to prescribe a drug are both basic decision problems. The first main result of this paper, Theorem \ref{basic} shows that given an updated set of DGPs\footnote{It is assumed to  always be a subset of the initial set of DGPs.}, the data-driven decisions are objectively better than the data-free decisions across \textit{all} basic decision problems if and only if the updated set of DGPs \textit{accommodates the true DGP}. By definition, a set accommodates a DGP if the convex closure of its marginals over future experiments \textit{contains} those of the DGP. This results thus emphasizes that accommodating the true DGP is the only crucial property for an inference rule to satisfy in order to robustly improve decisions over basic decision problems. 
	
	For general decision problems, i.e., choices among an arbitrary set of prospects, Theorem \ref{impossible} shows that it is impossible to guarantee the data-driven decisions to be always better when the updated set is a non-singleton set. This applies to the current setting as the true DGP may not be perfectly identified. However, when the updated set accommodates the true DGP, Theorem \ref{certaintyequivalent} shows that the data-driven decision can still guarantee an objective improvement upon the \textit{data-free certainty equivalent} across all decision problems. The certainty equivalent is defined as the expected utility under the worst DGP in the initial set. Thus, it is also the DM's initial subjective expectation from her data-free decision. Intuitively, this notion simply says that the DM would prefer to receive the data and make a data-driven decision rather than take the certainty equivalent for sure. This second notion of improvement, in addition, also characterizes accommodating the true DGP. Thus, this paper establishes a three-way equivalence: the data-driven decision objectively improves over data-free decision across all basic decision problems if and only if it objectively improves over data-free certainty equivalent across all decision problems, if and only if the updated set of DGPs accommodates the true DGP (Corollary \ref{three}). 
	
	Given this equivalence,  I then study updating rules in terms of accommodating the truth (true DGP). Through an illustrative example (Section \ref{example}), I first show that common updating rules such as maximum likelihood and Bayesian updating that perform well in the case of i.i.d., can in fact rule our the truth asymptotically almost surely in the presence of non-identical DGPs.\footnote{For Bayesian updating, this means that the posterior distribution asymptotically almost surely concentrates on DGPs other than the truth.} Thus by the previous result, they can lead to strictly worse decisions than without using the data. To address such an issue, I first propose a new updating rule, \textit{average-then-update}, and show that it guarantees to accommodate the truth asymptotically almost surely (Theorem \ref{asymptotic}) also with non-identical DGPs. 
	
	Furthermore, I study the case of finite sample by considering updating rules that guarantee to accommodate the truth with a pre-specified probability or confidence level. By definition, the updated sets are effectively consistent confidence regions of the true DGP. In the presence of non-identical DGPs, constructing confidence regions can be computationally challenging. To this end, the present paper proposes a tractable and easy-to-implement method. Specifically, given any sample data, the proposed rule first applies standard statistical techniques to construct a confidence region as if the DGPs are independent and identically distributed (i.i.d.). Then the updating rule retains DGPs in the initial set whose \textit{average of sample marginals},  i.e., the average mixture of marginal distributions over the sample data, equals the marginal of an i.i.d. distribution in the previous confidence region. Theorem \ref{confident} confirms that such an updating rule can guarantee the updated sets to accommodate the truth with at least the same confidence level as the i.i.d. distributions under any possible non-identical DGP in the initial set. 
	
	Finally, the decision framework studied in this paper is often applied to model economic problems such as dynamic portfolio choice, asset pricing, and social learning under model uncertainty and under ambiguity. For those problems, the existing literature obtains conclusions primarily by assuming that the DM applies full Bayesian updating.\footnote{Some may also refer to it as prior-by-prior updating.} However, the asymptotic result under full Bayesian updating is often hard to solve. In a commonly studied model of learning from Gaussian signals with ambiguous variances, I show that applying the average-then-update rule reduces to a simple and intuitive step. More importantly, applying the average-then-update rule also implies that learning is significantly more effective than full Bayesian updating (Proposition \ref{gaussian}). This learning outcome proves to be more intuitive in such a model. In addition, I provide a more concrete illustration of both proposed updating rules in Section \ref{Bernoulli} by studying a Bernoulli model with ambiguous nuisance parameters. There, I show that the updated sets have tractable expressions and intuitive interpretations. \\
	
	\textbf{Outline.} Section \ref{literature} reviews related literature. Section \ref{example} provides an illustrative example to demonstrate the issue with maximum likelihood and Bayesian updating. Section \ref{setup} introduces the formal decision environment. Section \ref{characterization} provides the characterization results that relate improving decisions to accommodating the true DGP. Section \ref{updating} formalizes the notions of accommodating the truth and provides the according updating rules. Section \ref{application} presents two applications of the proposed updating rules. Section \ref{conclusion} concludes. 
	
	\subsection{Related Literature}\label{literature}
	This paper formulates and studies a statistical decision problem where the sample data is arbitrarily drawn from a population of heterogeneous individuals. The independent but possibly non-identical assumption is meant to capture the fact that not only the individuals are heterogeneous but also the sampling procedure is arbitrary and unknown.\footnote{If the sampling is random, then the heterogeneity can be simply captured by an unobserved random variable with identical distribution across individuals.} This paper, therefore, contributes directly to the literature on statistical decision theory pioneered by \cite{wald1950statistical}. Especially, the present paper proposes a novel approach to this fundamental problem by considering whether or not data robustly improve decisions. Thus more specifically, it follows the literature on robust statistical decisions surveyed in \cite{Watson2016} and \cite{Hansen2016}. In addition, by developing inference methods for making statistical decisions, this paper is also closely related to papers in the econometrics literature. See \cite{manski2021econometrics} and \cite{Stoye2012} for a survey to that strand of work. 
	
	This paper borrows techniques from the literature on decisions under ambiguity to model this decision environment. Specifically, the environment studied in this paper is a direct generalization of that used in \cite{Epstein2007}. They assume the DM applies the maximum likelihood updating rule. The present paper highlights possible concerns with this approach. In a recent paper, \citet*{Epstein2016} develop robust confidence regions when the possible data-generating processes are belief functions. Belief functions impose restrictions on the possible marginal distributions. In contrast, the environment studied in the present paper allows for arbitrary marginal distributions. But the major conceptual difference from their paper, and other papers on asymptotic learning under ambiguity, such as \cite{Marinacci2002} and \cite{MARINACCI2019144}, is that the present paper emphasizes the implications of learning for decision-making. 
	
	This paper also contributes to the literature on dynamic decisions under ambiguity by proposing new updating rules for sets of distributions. See \citet*{gilboa_marinacci_2013} and \citet*{cheng2019relative} for recent developments. The essential departure of the present paper from papers in this literature is that it evaluates decisions according to an objective criterion, i.e., expected utility under the true DGP. This new approach allows me to establish a characterization of accommodating the true DGP property based on robustly improving decisions. This property is basically a generalized notion of statistical consistency. Thus, the present paper, for the first time in this literature, introduces statistical consistency to dynamic decisions under ambiguity. The objective approach also resonates with some recent papers that study misspecified learning according to objective criteria such as \citet*{he2020evolutionarily}, \citet*{Frick2021}, and references therein. 
	
	Finally, this paper develops useful techniques for making robust inferences in the presence of independent but non-identical distributions. The updated set of DGPs is evidently a set-valued identification objective, thus the idea here is related to the partial identification literature \citep*{Tamer2010, canay/shaikh:2017, MOLINARI2020355}. However in partial identification, the underlying DGPs are usually assumed to be only i.i.d.. Therefore, the inference results in this paper can potentially be useful in generalizing the results there to allow for non-identical distributions. In addition, this paper can also contribute to the literature  on making inferences from clustered data \citep*{HANSEN2019268} where the underlying distribution is often known to be non-identical. 
	
	\section{Illustrative Example}\label{example}
	In this section, I present an illustrative example to demonstrate how maximum likelihood and Bayesian updating can asymptotically almost surely lead to strictly worse decisions. Then I will briefly illustrate how using the average-then-update rule proposed in this paper can resolve this problem. 
	
	This example is given in the context of making content recommendation based on user feedback. Suppose Netflix has the option to recommend either a \textit{Standard} movie or a \textit{Personalized} movie to some users. Given the recommended movie, each user responds by either a \textit{like} or a \textit{dislike}, which gives Netflix a payoff of $1$ or $0$, respectively. For simplicity, also denote these outcomes by $\{1,0\}$. 
	
	Suppose Netflix identifies a group of users with similar online behaviors and profiles, but is uncertain about the outcomes of recommending either movie to those users. Thus for recommending a movie to those users, Netflix will face uncertainty over the state space given by the (possibly infinite) product of the two possible outcomes, denoted by $\{1,0\}^{\infty}$. 
	
	First, suppose there is no model uncertainty or ambiguity for recommending the Standard movie. It is known to generate the outcome like (or $1$) with the same probability of $1/2$ independently across all the users. Formally, the data-generating process (DGP) governing the uncertainty associated with recommending the Standard movie is thus known and given by the i.i.d. distribution with a marginal probability of observing a like equals $1/2$. I denote it by $(1/2)^{\infty}$. 
	
	On the other hand, suppose Netflix is ambiguous about recommending the Personalized movie and it only knows that there can be two different possibilities: 
	
	\begin{itemize}
		\item \textbf{Possibility \Rmnum{1}:} The Personalized content is a good match with the users so that it induces each user to like it independently with a probability of at least $0.6$. However, due to unobservable heterogeneity across the users, such as their offline activities, the specific probability of observing a like from different users may not be the same. Given that the heterogeneous users may interact with Netflix at any point in time, Netflix's initial knowledge is therefore represented by the following set of DGPs: 
		\begin{equation*}
			[0.6, 1]^{\infty} \equiv \{P \in \Delta_{indep}(\{1,0\}^{\infty}) : P_{i} \in [0.6, 1], \forall i\},
		\end{equation*}
		where $\Delta_{indep}(\{1,0\}^{\infty})$ denotes the set of all probability distributions over $\{1,0\}^{\infty}$ that are independent across users, and $P_{i}$ denotes the marginal probability of observing a like from the $i$-th user under $P$. It is important to notice that this set includes all independent but non-identical distributions whose marginal probabilities of $1$ can be different across users but are always between $0.6$ and $1$. 
		
		\item \textbf{Possibility \Rmnum{2}:} The Personalized movie is a mismatch with the users. For simplicity, suppose in this case it  generates the outcome like with the same probability of $1/3$ independently across all the users. The DGP governing the uncertainty here is thus denoted by $(1/3)^{\infty}$. 
		
	\end{itemize}
	
	Taking both possibilities into account, Netflix then considers that the union, $[0.6, 1]^{\infty}\cup\{(1/3)^{\infty}\}$, gives the initial set of DGPs that represents the uncertainty associated with recommending the Personalized movie. Notice that the minimum expected payoff from recommending the Personalized movie to a user is $1/3$, strictly less than the expected payoff from recommending the Standard movie ($1/2$). Thus, Netflix's data-free decision will be to recommend the Standard movie to the users. 
	
	Next, suppose Netflix observes historical outcomes from having recommended Personalized movie to the users within this group. With the sample data, a common inference method is to look at the likelihood of observing the data under different DGPs. Formally, applying the maximum likelihood updating \citep{GILBOA199333} in this example is to retain only those DGPs in the initial sets that maximizes the likelihood of observing the data. 
	
	Consider the case where the true DGP governing the Personalized movie is $(1/3)^{\infty}$, i.e. the Personalized movie is a mismatch with the users. Then by the law of large numbers, almost surely, Netflix will eventually observe some sample data with an empirical frequency of like close to $1/3$. The critical observation here is that the true DGP $(1/3)^{\infty}$ does not maximize the likelihood of observing such data among all the possible DGPs in the initial set. 
	
	To be more specific, consider any sequence of realized outcomes with an empirical frequency of like exactly equals $1/3$. Given such a sequence, one can always find the DGP in the initial set whose marginal probability of observing a like from a user is $1$ if the realized outcome from that user is a like. Similarly, its marginal probability of observing a like from a user is $0.6$ if the realized outcome is a dislike. Notice that such a DGP maximizes the likelihood of observing the realized sequence among all the possible DGPs in the set $	[0.6, 1]^{\infty} $.  In addition, one can verify the following inequality is true: 
	\begin{equation*}
		(1/3)^{1/3}\times (2/3)^{2/3}  <  (1)^{1/3} \times (0.4)^{2/3}.
	\end{equation*}
	This inequality thus implies that the likelihood of observing this data under the true DGP (the left-hand side) is strictly less than the likelihood under the DGP just described (the right-hand side). Therefore, for any sample data with an empirical frequency of likes close to $1/3$, maximum likelihood updating will rule out $(1/3)^{\infty}$ from the updated set. As a result, when the true DGP is $(1/3)^{\infty}$, the updated sets under maximum likelihood updating will rule out the true DGP asymptotically almost surely.\footnote{While the present paper is, to the best of my knowledge, the first to make such an observation, it is intrinsically related to the infamous incidental parameter problem with using maximum likelihood to make estimations discovered by \cite{Neyman1948}. See \cite{LANCASTER2000391} for a review.} 
	
	Furthermore, notice that because all the possible DGPs are independent, data-driven decisions are determined by data solely through the updated set of DGPs.\footnote{The conditional marginal distributions over the future realizations are always the same as the unconditional marginal distributions.} When $(1/3)^{\infty}$ is ruled out from the updated set, the minimum expected payoff from recommending the Personalized movie according to the updated set is $0.6$, strictly higher than recommending the Standard movie. Therefore, the data-driven decision under maximum likelihood updating will asymptotically almost surely be to recommend the Personalized movie when the true DGP is $(1/3)^{\infty}$. However, the responses from future users are still governed by the true DGP. In other words, objective payoff from the data-driven decision will be $1/3$ from each user, strictly lower than what Netflix can get from simply ignoring the data and taking the data-free decision ($1/2$). The same conclusion also applies more generally to updating rules that based on maximum likelihood, such as relative maximum likelihood \citep*{cheng2019relative} and likelihood-ratio updating \citep*{Epstein2007}. 
	
	More importantly, a similar observation also applies when Netflix applies Bayesian updating. In order to apply Bayesian updating, Netflix needs to form a prior distribution over the initial set of DGPs. Without loss of generality, when Netflix uses a uniform prior, under Bayesian updating, the posterior probability of a DGP will then be proportional to the likelihood of observing the realized data. Therefore, similar to the case of maximum likelihood, the posterior probability of the true DGP will also be asymptotically driven to zero.\footnote{See Appendix \ref{examplediscussion} for a detailed discussion on how this conclusion is obtained and also for a discussion about what if Netflix instead uses a Bayesian decision criterion.} Although Bayesian updating is known to eventually almost surely concentrate on the true DGP when the underlying DGPs are all i.i.d. \citep{miller2018detailed}, this observation implies that this nice property in general cannot be guaranteed when the underlying DGPs can be non-identical.\footnote{\cite{choi2008remarks} talk about Bayesian consistency in the presence of independent but non-identical distributions. Their necessary and sufficient condition does not hold in this example. }
	
	Finally, another commonly suggested updating rule for sets of distributions is full Bayesian updating \citep{pires2002rule}. By definition, under full Bayesian updating, the updated set will retain every possible DGP.\footnote{Different from Bayesian updating, it does not require a prior distribution over the DGPs, but each DGP is the ``prior'' that will be updated using Bayes' rule conditional on the data received. Here, because of independence, applying Bayes' rule does not change the marginal distributions over the future events.} As a result, the data-driven decision is always the same as the data-free decision. In other words, full Bayesian updating is equivalent to simply ignoring the data in this example. This is also true in the general decision environment. 
	
	So far, I have shown that all the existing updating rules have undesirable features in terms of using data to improve decisions. I will now illustrate how the proposed updating rule, \textit{average-then-update}, can asymptotically almost surely lead to strictly better decisions. The intuitive idea of the average-then-update rule is that it retains all the DGPs in the initial set whose average of sample marginals is close enough to the \textit{empirical distribution} of the realized outcomes.\footnote{The formal definition is that the sup norm distance is less than some pre-specified $\epsilon > 0$.} For example, when observe data with an empirical frequency of like equals $0.8$, applying the average-then-update rule will retain, for instance, the i.i.d. distribution $(0.8)^{\infty}$ and non-identical distribution  $(0.6\times 1)^{\infty}$ whose marginals are $0.6$ for the odd users and $1$ for the even users. 
	
	The key intuition for applying this updating rule in this example is the fact that, under any independent but non-identical distribution, the empirical distribution converges to this average almost surely. In turn, it ensures that the true DGP will be asymptotically almost surely preserved in the updated set. 
	
	In terms of decisions, notice that when the true DGP is $(1/3)^{\infty}$, the data-driven decision under the average-then-update rule will asymptotically almost surely be to recommend the Standard movie. Thus, the objective payoff under the data-driven and data-free decisions will be the same. On the other hand, if the true DGP belongs to $	[0.6, 1]^{\infty} $, then asymptotically almost surely $(1/3)^{\infty}$ will be ruled out from the updated set, and the data-driven decision will be to recommend the Personalized movie. Notice that in this case, objective payoff from the data-driven decision will be strictly higher than the data-free decision. 
	
	In summary, by applying the average-then-update rule, Netflix's objective payoff from data-driven decisions in this example will asymptotically almost surely be either the same or strictly higher compared to its data-free decision. In other words, data indeed robustly improve decisions. Next, this paper moves beyond this example by formalizing the general decision environment and shows how this observation can be generalized.
	
	\section{Decision Environment}\label{setup}
	\subsection{State Space, Data, and Decision Problems}
	The state space is given by a countably infinite sequence of \textbf{random experiments}. They are ordered and indexed by the set $\N = \{1,2,\cdots\}$. Each experiment yields an outcome in a finite set $S$ with generic outcome denoted by $s_{j}$. Let $S_{i}$ denote the set of possible outcomes for the $i$-th experiment, although $S_{i} = S$ for all $i$. The full state space is $\Omega = S^{\infty}$. Let $\Sigma$ denote the discrete sigma-algebra on $S$ and $\Sigma^{\infty}$ the corresponding product sigma-algebra on $\Omega$. 
	
	For any sample size $N \in \N$, let $S_{N} \equiv \prod_{i=1}^{N}S_{i}$ denote the \textbf{sample experiments}. The \textbf{sample data} $\omega_{N} \in S_{N}$ thus are simply given by realizations of the sample experiments. Similarly, let $S^{N} \equiv \prod_{i=N+1}^{\infty} S_{i} $ denote \textbf{future experiments} and they will only realize after the decision-maker (DM) makes a decision. Similarly, I use $\Sigma_{N}$ and $\Sigma^{N}$ to denote the corresponding product sigma algebras. 
	
	The DM makes a decision by choosing an \textbf{act}. An act, denoted by$f: S^{N} \rightarrow [0,1]$, is formally a $\Sigma^{N}$-measurable function that maps outcomes of future experiments to utilities represented by real numbers in $[0,1]$. Let $\mathcal{F}^{N}$ denote the space of all such acts endowed with the product topology. 
	
	More specifically, the DM will face a \textbf{decision problem} $D^{N}$, which is defined as a compact subset of the acts, i.e., $D^{N} \subseteq \mathcal{F}^{N}$. In other words, the DM's decision is more formally the choice from a given decision problem. Let $\mathcal{D}^{N}$ denote the collection of all decision problems. 
	
	In addition, a special type of decision problem that will be important for the analysis of this paper is the \textbf{basic decision problem}. It is defined as a decision problem consisting of only two acts and one of them is constant, i.e., pays the same utility in all states. Let $D^{N}_{b} = \{f,x\}$ denote a generic basic decision problem and let $x \in [0,1]$ denote the constant act. Similarly, let $\mathcal{D}_{b}^{N}$ denote the collection of all basic decision problems. 
	
	\subsection{Data-Generating Process}
	The uncertainty over the experiments is governed by a \textbf{data generating process (DGP)}, a sequence of independent but possibly non-identical distributions over $\Omega$. Let $\Delta_{indep}(\Omega)$ denote the set of all countably additive probability measures over $\Omega$ that are independent across experiments. Let $P_{i}$ denote its marginal distribution over the $i$-th experiment $S_{i}$. Furthermore, for any $P$, let $P_{N}$ and $P^{N}$ denote its marginals over sample and future experiments, i.e. over $S_{N}$ and $S^{N}$ respectively.\footnote{To make a distinction, whenever referring to the marginal distribution over the $N$-th experiment, I will use $P_{i = N}$.} Then, by definition, for any $P \in \Delta_{indep}(\Omega)$, for all $N$, $\omega_{N}$, and $E \in \Sigma^{N}$, if $P(\omega_{N}) > 0$ then $P(E|\omega_{N}) = P^{N}(E)$. Equip $\Delta_{indep}(\Omega)$ with the weak* topology. 
	
	The DM's initial knowledge about the underlying uncertainty is represented by an \textbf{initial set}, which is a compact subset of $\Delta_{indep}(\Omega)$ denoted by $\mathcal{P}$. Furthermore, assume that all the probability measures in the initial set have full support.\footnote{All the results can be adjusted accordingly without the full-support assumption. Furthermore, it is a standard assumption in updating sets of distributions because full Bayesian updating is rigorously defined and characterized only when it holds. } Let $\mathcal{P}_{i}$, $\mathcal{P}_{N}$ and $\mathcal{P}^{N}$ denote the corresponding sets of marginals. The following primitive assumption simply says that every possible DGP in the initial set can be the true one that governs the underlying uncertainty. 
	\begin{assumption}\label{assumption1}
		$\mathcal{P}$ is the set of all possible DGPs. 
	\end{assumption}
	This assumption is a counterpart of the ``grain of truth'' assumption in the Bayesian learning literature \citep{Kalai1993}. Although it is less demanding in the sense that it only requires the DM to include the possible DGPs in the initial set but does not have to form any probabilistic assessment. While every DGP in the initial set can be the true one governing the underlying uncertainty, I will use $P^{*}$ to denote the DGP that governs the data in context to distinguish it from other DGPs in the initial set. 
	
	\subsection{Decision Rules and Payoffs}
	Given a decision problem $D^{N}$, the DM can make a \textbf{data-free decision} based only on her initial knowledge about the underlying uncertainty. Let $c(D^{N})$ denote the DM's data-free decision, which represents her choice from $D^{N}$ according to the maxmin expected-utility (MEU) criterion according to the initial set of DGPs. Formally, 
	\begin{equation*}
		c(D^{N}) \equiv \arg\max\limits_{f \in D^{N}}\min\limits_{P \in \mathcal{P}} \int_{S^{N}} f(\omega)dP^{N}(\omega), 
	\end{equation*}
	where $c(D^{N})$ is always a singleton given by some arbitrary tie-breaking rule. 
	
	With sample data $\omega_{N}$, the DM can also use the data to update the initial set of DGPs to an \textbf{updated set} of DGPs. Let $\mathcal{P}(\omega_{N}) \subseteq \Delta_{indep}(\Omega)$ denote the updated set, which depends on only the initial set $\mathcal{P}$ and data $\omega_{N}$. Importantly, it cannot depend on the specific decision problem. This is the standard \textit{consequentialist} property of an updating rule.\footnote{For example, both maximum likelihood and Bayesian updating are consequentialist. For the exact definition and an extensive discussion, see \citet{hanany2007updating, hanany2009updating} and \citet{Siniscalchi2009-SINTOO-5}.} Moreover, suppose the updated set satisfies the following primitive assumption: 
	\begin{assumption}\label{assumption3}
		$\mathcal{P}(\omega_{N}) \subseteq \mathcal{P}$.
	\end{assumption}
	This assumption should be intuitive given Assumption \ref{assumption1}. If the initial set always contains the true DGP, then there is no reason to consider additional distributions in updating. 
	
	The \textbf{data-driven decision} is then defined as the MEU decision according to the updated set of DGPs. Formally, let $c(D^{N}, \omega_{N})$ represent the DM's choice from $D^{N}$ conditional on $\omega_{N}$: 
	\begin{equation*}
		c(D^{N}, \omega_{N}) \equiv \arg\max\limits_{f \in D^{N}}\min\limits_{P \in \mathcal{P}(\omega_{N})} \int_{S^{N}} f(\omega)dP^{N}(\omega),
	\end{equation*}
	where $c(D^{N}, \omega_{N})$ is also always a singleton given by any tie-breaking rule that satisfies the following consistency assumption: 
	
	\begin{assumption}\label{assumption4}
		If $c(D^{N})$ also maximizes the minimum expected-utility according to the updated set, then $ c(D^{N},\omega_{N}) = c(D^{N})$. Otherwise, the tie-breaking can be arbitrary. 
	\end{assumption}
	
	This consistency assumption has a convenient implication: If the DM's updated set coincides with the initial set, then her data-driven and data-free decisions will always be the same. This rules out the uninteresting complication that arises when the DM's decisions are different purely because of the tie-breaking rules. 
	
	The DM's \textbf{objective payoff} from her decision will be determined by the true DGP that governs the future experiments. To make the dependence of the true DGP more explicit, I use $W(f, P^{*})$ to denote objective payoff from act $f$ when the true DGP is $P^{*}$, i.e., 
	
	\begin{equation*}
		W(f, P^{*}) = \int_{S^{N}} f(\omega) dP^{*N}(\omega).
	\end{equation*}
	For any decision problem $D^{N}$, the DM's objective payoffs from her data-free and her data-driven decisions are then denoted by $W(c(D^{N}), P^{*}) $ and $W(c(D^{N},\omega_{N}), P^{*})$, respectively. 
	
	Because the decisions depend only on the future experiments and the DM applies the MEU criterion, any set of DGPs will lead to the same decisions as the closed and convex hull of the set of its marginals over future experiments. As a result, when talking about the relation between DGPs and sets of DGPs, I will generalize the notion of containing accordingly. For any set, let $co(\cdot)$ denote its closed and convex hull. Say that the updated set $\mathcal{P}(\omega_{N})$ \textbf{accommodates} the DGP $P^{*}$ if $P^{*N} \in co(\mathcal{P}(\omega_{N})^{N})$. Similarly, the updated set $\mathcal{P}(\omega_{N})$ \textbf{refines} the initial set $\mathcal{P}$ if $co(\mathcal{P}(\omega_{N})^{N}) \subsetneqq co(\mathcal{P}^{N})$. 
	
	Finally, an additional comment about the definitions of the decision rules is that they are defined as choices from $D^{N}$ but not from all the probability distributions over $D^{N}$. This definition is, in fact, more general because adding the probability distributions into $D^{N}$ would be another well-defined decision problem. Thus, the current setup allows for richer decision patterns. All the results in this paper can also be obtained using the other definition. Moreover, the current definition can also be interpreted by the phenomenon that the DM does not believe her own randomization can hedge against ambiguity as characterized in \cite{Saito2015} and \cite{Ke2020}. 
	
	\section{Characterization Results}\label{characterization}
	In this section, I study a ``static'' comparison of data-free and data-driven decisions by fixing some true DGP $P^{*}$ and sample data $\omega_{N}$. Specifically, given some updated set $\mathcal{P}(\omega_{N})$, the problem simply reduces to the comparison between two MEU decisions, one under $\mathcal{P}$ and one under $\mathcal{P}(\omega_{N})$,  according to the expected utility under some $P^{*} \in \mathcal{P}$. 
	
	\subsection{Basic Decision Problems}
	Recall in the illustrative example, Netflix faces a basic decision problem given by the choice between recommending a Standard movie (constant act) and a Personalized movie (ambiguous act). For this specific decision problem, the illustrative example has shown that if the DM's updated set rules out the true DGP (maximum likelihood updating), then the data-driven decision will be objectively worse than the data-free decision. In contrast, if the updated set contains the true DGP (average-then-update), then the data-driven decision will be objectively better. 
	
	This subsection first investigates whether this observation can be generalized to across all basic decision problems. To formalize this idea, consider the following definition. 
	
	\begin{definition}
		Given an updated set $\mathcal{P}(\omega_{N})$, the data-driven decision objectively dominates the data-free decision across \textbf{all basic decision problems} under the DGP $P^{*}$ if for all $D^{N}_{b} \in \mathcal{D}^{N}_{b}$, 
		\begin{equation*}
			W(c(D^{N}_{b}, \omega_{N}), P^{*}) \geq W(c(D^{N}_{b}), P^{*}).
		\end{equation*}
		The dominance is strict if there exists $D_{b}^{N}$ such that the strict inequality holds. 
	\end{definition}
	
	When this definition holds, then the DM will be certain that in any basic decision problem, her data-driven decision can always guarantee a higher objective payoff. Notice that while the DM makes decisions considering the worst-case distribution, what this definition ensures is the improvement under the \textit{true} DGP, which is not necessarily the worst-case. The key to such an improvement, as shown in the following theorem, is precisely that the updated set accommodates the true DGP. 
	
	\begin{theorem}\label{basic}
		Given an updated set $\mathcal{P}(\omega_{N})$, the data-driven decision objectively dominates the data-free decision across all basic decision problems under the DGP $P^{*}$ if and only if $\mathcal{P}(\omega_{N})$ accommodates $P^{*}$. Moreover, the dominance is strict if and only if the updated set $\mathcal{P}(\omega_{N})$ refines the initial set $\mathcal{P}$. 
	\end{theorem}
	
	The proof of this theorem relies on the following observation. Whenever the data-free and data-driven decisions are different for a basic decision problem, it needs to be the case that the data-free decision chooses the constant act and the data-driven decision chooses the ambiguous act. This is further implied by Assumption \ref{assumption3}, i.e., $\mathcal{P}(\omega_{N})$ is always a subset of $\mathcal{P}$. Then the ``if'' direction is shown by noticing that objective payoff from the ambiguous act is weakly higher than the minimum expected payoff under the updated set, which by presumption is also weakly higher than the payoff from the constant act. 	
	
	Intuitively, for a basic decision problem, choosing the ambiguous act over the constant act reflects confidence in knowledge about the underlying uncertainty. As data necessarily enhance such confidence (for the updated set is always a subset), accommodating the true DGP requires the confidence to be enhanced towards a correct direction. In terms of decisions, this correct direction further implies that choosing the ambiguous act in this case is always objectively profitable. \\
	
	\textbf{Remark.} This characterization result crucially relies on the DM applies the maxmin expected-utility criterion. Appendix \ref{regret} presents an example where if the DM applies the minimax regret criterion\footnote{For the exact definition and reference, see \citet{STOYE20112226}.}, then even when the updated set accommodates the true DGP, objective payoff can still be strictly higher from the data-free decision. This observation, on the other hand, provides a new motivation for using the maxmin expected-utility criterion for making robust decisions. 
	
	\subsection{General Decision Problems}
	While Theorem \ref{basic} provides a characterization for basic decision problems, it does not directly provide any implication for non-basic decision problems. Given Theorem \ref{basic}, one might be tempted to come up with the following intuition: given the updated set is closer to the true DGP than the initial set in the sense that it refines the initial set and also accommodates the true DGP, it is intuitive that the maxmin decisions under the smaller set would be always better than under the larger set. In this subsection, I will show that, however, such an intuition cannot be generalized to across \textit{all} decision problem. 
	
	Similarly, consider the following definition: 
	\begin{definition}
		Given an updated set $\mathcal{P}(\omega_{N})$, the data-driven decision objectively dominates the data-free decision across \textbf{all decision problems} under the DGP $P^{*}$ if for all $D^{N} \in \mathcal{D}^{N}$, 
		\begin{equation*}
			W(c(D^{N}, \omega_{N}), P^{*}) \geq W(c(D^{N}), P^{*}).
		\end{equation*}
		The dominance is strict if there exists $D^{N}$ such that the strict inequality holds. 
	\end{definition}
	
	First notice that this definition is feasible when the true DGP is uniquely identified from data, i.e., when $\mathcal{P}(\omega_{N}) = \{P^{*}\}$. In this case, the data-driven decision becomes the exact optimal decision under $P^{*}$, and therefore objective dominance across all decision problems holds. 
	
	In the current environment, however, the true DGP may only be \textit{partially identified} from any possible data. Because when the possible DGPs are independent but non-identically distributed, observing only one realization from each experiment may not be sufficient to distinguish between different DGPs. More importantly, even in a rare case where the DM can perfectly observe the marginal distribution $P^{*}_{N}$, there may still be multiple possible DGPs in the initial set that have the same marginal over sample experiments but different marginals over the future.\footnote{In the illustrative example, this is true when $P^{*} \in [0.6, 1]^{\infty}$.} As a result, the updated set may at best be a non-singleton set, and thus the true DGP can only be partially identified. 
	
	Then in this case, objective dominance is meaningful only when it can hold under multiple DGPs, because any one of them may be the one that governs future experiments. However, the following theorem asserts that it is impossible if data ever changes decision (if data does not change the decision at all, then it is trivially true with always equalities). 
	
	\begin{theorem}\label{impossible}
		Given an updated set $\mathcal{P}(\omega_{N})$, if it refines the initial set $\mathcal{P}$, then objective dominance across all decision problems can hold under at most one DGP (in terms of marginal over future experiments). 
	\end{theorem}
	
	For an intuition of this theorem, I refer directly to the proof in Appendix \ref{impossible_proof}. This theorem implies that, if the true DGP is only partially identified, then the DM can never use the data in a way to guarantee her data-driven decisions to be objectively better across all decision problems under multiple possible DGPs. Moreover, it also suggests that, even when the updated set is closer to the true DGP than the initial set, there still exist decision problems in which the maxmin decision is objectively better under the larger set.\footnote{See Lemma \ref{objectivedominancelemma} in Appendix \ref{impossible_proof} for the precise statement and intuition.} The previous intuition, therefore, applies to only basic decision problems. As a result, for a non-basic decision problem, the data-driven and data-free decisions may be incomparable according to the possible objective payoffs. 	
	
	Although objectively incomparable, the DM can still evaluate whether or not she is \textit{subjectively} willing to take the data-driven decision instead of ignoring the data. Notice that by taking the data-free decision, the DM guarantees herself a \textit{data-free certainty equivalent} in any decision problem. Specifically, no matter which possible DGP governs the uncertainty, objective payoff from her data-free decision is always higher than the minimum expected utility across DGPs in the initial set. This minimum is also the DM's subjective expectation under the maxmin expected-utility criterion. Therefore, the DM would be subjectively willing to take a data-driven decision only if she knows that objective payoff from data-driven decisions would be greater than this certainty equivalent. 
	
	The following definition formalizes this idea: 
	\begin{definition}
		Given an updated set $\mathcal{P}(\omega_{N})$, the data-driven decision  \textbf{objectively improves upon the data-free certainty equivalent} across all decision problems under the DGP $P^{*}$ if for all $D^{N} \in \mathcal{D}^{N}$, 
		\begin{equation*}
			W(c(D^{N}, \omega_{N}), P^{*}) \geq \min\limits_{P \in \mathcal{P}} W(c(D^{N}), P).
		\end{equation*}
	\end{definition}
	
	Notice that if such an improvement holds, whenever the DM is offered a choice between receiving the data or getting her data-free certainty equivalent, she will always prefer to receive the data. In other words, data-driven decision is always acceptable. The next result shows that, in fact, this notion of improvement is also necessary and sufficient to that the updated set accommodates the true DGP. 
	
	\begin{theorem}\label{certaintyequivalent}
		Given an updated set $\mathcal{P}(\omega_{N})$, the data-driven decision objectively improves upon the data-free certainty equivalent across all decision problems under the DGP $P^{*}$ if and only if $\mathcal{P}(\omega_{N})$ accommodates $P^{*}$. 
	\end{theorem}
	
	To give a sketch of the proof, first observe that when the updated set accommodates the DGP $P^{*}$, objective payoff from the data-driven decision under $P^{*}$ is always higher than the minimum expected utility across all DGPs in the updated set. Then the conclusion is obtained by showing that this minimum always dominates the data-free certainty equivalent. 
	
	Importantly, combining both Theorem \ref{basic} and Theorem \ref{certaintyequivalent} gives the whole picture of the characterization result: 
	\begin{corollary}\label{three}
		Given an updated set $\mathcal{P}(\omega_{N})$ and DGP $P^{*}$, the followings are equivalent: 
		\begin{enumerate}[(i)]
			\item $\mathcal{P}(\omega_{N})$ accommodates $P^{*}$. 
			\item The data-driven decision objectively dominates the data-free decision across all basic decision problems under $P^{*}$.
			\item The data-driven decision objectively improves upon the data-free certainty equivalent across all decision problems under $P^{*}$.
		\end{enumerate}
	\end{corollary}
	
	This result highlights the equivalence between the two different notions of improvement across basic and across general decision problems: While the two notions are different both in the form and in the domain where they apply, this characterization result asserts that any one of them always implies the other. 
	
	In addition, it has been emphasized that when the updated set accommodates the true DGP, in a non-basic decision problem, it is still possible that objective payoff is higher from the data-free decision. However, when the data-free and data-driven decisions are different, it is impossible to have the data-free decision to be objectively better under \textit{every} DGP in the updated set.\footnote{Consider when the true DGP is the worst-case DGP in the updated set for the data-driven decision. The objective payoff under the data-driven decision is then strictly higher than the data-free decision, because by Assumption \ref{assumption4}, the data-free decision is not optimal given this worst-case DGP.} Thus, when the updated set accommodates the true DGP, there will never be a decision problem in which the DM finds data-free decision uniformly preferable, i.e.,  it is objectively better under all DGPs in the updated set. 
	
	In summary, knowing the updated set accommodates the true DGP, for basic decision problems, the data-driven decision is objectively better under all possible DGPs. When the decision problem is non-basic, the data-driven decision first guarantees at least the data-free certainty equivalent; furthermore, there always exists a possible DGP under which the data-driven decision is strictly better than the data-free decision. 
	
	\section{Updating Rules and Accommodating the Truth}\label{updating}
	The previous section shows that using data to robustly improve decisions is equivalent to accommodate the true data-generating process. This section turns to the study of updating rules in terms of how to accommodate the truth (true DGP) based on sample data. 
	
	An \textbf{updating rule} is formally a mapping from the initial set $\mathcal{P}$ and sample data $\omega_{N}$ to an updated set $\mathcal{P}(\omega_{N})$. When data is stochastically generated by some DGP, an updating rule determines \textit{how often} the updated sets, which become random sets, accommodate the DGP in probabilities. 
	
	One possible requirement is to let the updated sets accommodate the truth \textbf{everywhere}, i.e., under every possible DGP $P^{*}$, the updated sets accommodate it for every possible sample data $\omega_{N}$. A straightforward observation summarized by the following proposition is that it is true if and only if the updated sets never refine the initial set.
	
	\begin{proposition}
		The updated sets accommodate the truth everywhere if and only if $	co(\mathcal{P}(\omega_{N})^{N}) = co(\mathcal{P}^{N})$ for all $\omega_{N}$. 
	\end{proposition}
	
	This condition further implies the data-driven and data-free decisions are always the same. In other words, it is possible only when the DM ignores the data or applies the full Bayesian updating. Thus, accommodating the truth everywhere is too strong to be meaningful. 
	
	In the following, depending on whether or not the sample size can go to infinity, I consider two weaker notions of accommodating the truth and provide new updating rules that guarantee these notions. 
	
	\subsection{Infinite Sample}
	When the sample size can increase unboundedly to infinity, the first notion of accommodating the truth is given by an asymptotic sense. Consider the following definition. 
	
	\begin{definition}
		The updated sets \textbf{accommodate the truth asymptotically almost surely} if for all $P^{*} \in \mathcal{P}$ and for $P^{*}$-almost every $\omega \in \Omega$, there exists $\bar{N}(\omega)$ such that for all $N \geq \bar{N}(\omega)$, the updated set $\mathcal{P}(\omega_{N})$ accommodates the DGP $P^{*}$. 
	\end{definition}
	
	Intuitively, it says that under any possible DGP $P^{*}$ in the initial set, when the sample size is sufficiently large, the updated sets will almost surely accommodate the true DGP that generates the data. Given the characterization result, this definition also implies that the data-driven decisions will asymptotically almost surely be better than the data-free decisions. Recall from the illustrative example, maximum likelihood updating is shown to violate this criterion, and thus leads to strictly worse decisions almost surely. 
	
	In this paper, I propose a new updating rule, \textit{average-then-update}, that can be shown to satisfy this criterion in the current environment. Before providing the formal definition, some additional notations are needed. 
	
	For sample data $\omega_{N}$, let $\boldsymbol\Phi(\omega_{N}) \in \Delta(S)$ denote the \textbf{empirical distribution}, i.e., for any outcome $s_{j} \in S$, 
	\begin{equation*}
		\boldsymbol\Phi(\omega_{N})(s_{j})  \equiv  N^{-1}\sum\limits_{i=1}^{N} I\{\omega_{i} = s_{j}\}.
	\end{equation*}
	
	For any $P \in \Delta_{indep}(\Omega)$ and for any $N$, the \textbf{average of sample  marginals}, $\bar{P}_{N} \in \Delta(S)$,  is defined to be the distribution over $S$ given by the average mixture of marginal distributions over sample experiments, i.e.,  
	\begin{equation*}
		\bar{P}_{N} \equiv N^{-1}\sum_{i=1}^{N}P_{i},
	\end{equation*}
	where for each $i$, $P_{i} \in \Delta(S)$ is the marginal distribution over the $i$-th experiment. For any $p,q\in \Delta(S)$, let $\rho(p,q)$ denote the sup-norm distance.\footnote{As $S$ is finite, the choice of distance is not essential.} The average-then-update rule is then formally defined by the following. 
	
	\begin{definition}
		The updated sets are given by \textbf{average-then-update} if for some pre-specified $\epsilon > 0$ and for all $\omega_{N}$,
		\begin{equation*}
			\mathcal{P}(\omega_{N})  = \{P\in \mathcal{P}: \rho (\bar{P}_{N}, \boldsymbol\Phi(\omega_{N}) ) < \epsilon \}.
		\end{equation*}
	\end{definition}
	
	The average-then-update rule is based on the simple heuristic of retaining a DGP if its average of sample marginals is close enough to the empirical distribution. Such a heuristic is often used when the possible DGPs are all i.i.d. for it corresponds to maximum likelihood updating in that case. When the possible DGPs can be non-identical, while maximum likelihood updating turns out to be no longer useful, the following theorem shows that this heuristic remains valid. 
	
	\begin{theorem}\label{asymptotic}
		The updated sets given by average-then-update with any $\epsilon > 0$ accommodate the truth asymptotically almost surely. 
	\end{theorem}
	
	The key intuition of this theorem relies on Kolmogorov's strong law of large numbers, which further implies that under any sequence of independent but possibly non-identical distributions, the empirical distribution converges to the average of sample marginals in the sup-norm almost surely. This intuition also implies that the independence assumption is not crucial for this result. When the possible DGPs might have correlations between experiments, as long as there is a version of the strong law of large numbers that applies, the same conclusion will still hold. 
	
	Moreover, average-then-update rule can refine the initial set as shown in the illustrative example. Therefore, it can also lead to strictly better data-driven decisions asymptotically almost surely. 
	
	\subsection{Finite Sample}
	When the sample size is only finite, the relevant criterion for accommodating the truth is with some pre-specified probability or confidence level. 
	
	\begin{definition}
		The updated sets \textbf{accommodate the truth with an asymptotic level $1-\alpha$} if for all $P^{*} \in \mathcal{P}$, 
		\begin{equation*}
			\liminf_{N \rightarrow \infty} P^{*} (\{\omega_{N}: \mathcal{P}(\omega_{N}) \text{ accommodates } P^{*}\}) \geq 1-\alpha.
		\end{equation*}
	\end{definition}
	
	In terms of decisions, this definition implies that with any sample data, the data-driven decision will be objectively better than the data-free decision with at least the asymptotic probability $1-\alpha$. 
	
	Given this definition, the updated sets are effectively consistent \textit{confidence regions} of the true DGP, although the coverage is in a weaker sense of accommodating.\footnote{This difference is not important for the analysis that follows.} Thus finding an updating rule that guarantees this criterion would be equivalent to constructing the corresponding confidence regions. Constructing confidence regions is theoretically straightforward by exploiting the well-known duality between confidence regions and hypothesis tests.
	
	Specifically, for any possible DGP $P \in \mathcal{P}$, one can use the sample data to test the following null hypothesis 
	\begin{equation*}
		H_{P}: P^{*} = P
	\end{equation*}
	versus the unrestricted alternative hypothesis $P^{*} \neq P$. Formally, the test will be defined by the \textit{region of acceptance} denoted by $E_{N, \alpha}(P) \subseteq S_{N}$. Namely, when the sample data $\omega_{N}$ belongs to the region of acceptance, then the null hypothesis will be accepted or not rejected by the data. In addition, suppose the regions of acceptance satisfy the following condition:
	\begin{equation*}
		\liminf_{N \rightarrow \infty} P(E_{\alpha,N}(P)) \geq 1-\alpha.
	\end{equation*}
	It simply says that the probability of accepting the null hypothesis when it is true will be asymptotically greater than $1-\alpha$. 
	
	Given any data $\omega_{N}$, let the updated set be the set of DGPs that are accepted by the data under each one of the individual hypothesis tests: 
	\begin{equation*}
		\mathcal{P}(\omega_{N}) = \{P \in \mathcal{P}: \omega_{N} \in E_{\alpha, N}(P)\}.
	\end{equation*}
	Then it is straightforward to show that such an updated set contains the true DGP and thus accommodates it with the asymptotic level $1-\alpha$. 
	
	As a result, constructing confidence regions simply boils down to finding the regions of acceptance for every possible DGP and every sample size. It is a manageable task when the possible DGPs are all i.i.d. as it is the standard statistical procedure for constructing confidence regions. More fundamentally, it is so because there are two convenient features for i.i.d. DGPs. By the central limit theorem, the regions of acceptance can be constructed using the probability contours of the corresponding multivariate Gaussian distributions. In the case of i.i.d., conveniently, the mean vector and covariance matrix of the corresponding Gaussian distribution depend only on the marginal distribution thus do not change with the sample size. Furthermore, as each i.i.d. distribution is uniquely determined by its marginal distribution, the number of tests is also fixed regardless of the sample size.  	
	
	Both of the convenient features no longer hold when the possible DGPs can be non-identically distributed. First, in this case, both the mean vector and covariance matrix depend on all the marginal distributions over the sample experiments. As a result, for every different sample size, both need to be recalculated even for the same DGP. Moreover, non-identical DGP is determined by every one of its marginal distributions. Thus, every additional realization can potentially increase the number of possible DGPs that need to be tested by the data. Therefore, in the presence of non-identical distributions, directly finding the regions of acceptance for every DGP may be too complicated to be feasible. 
	
	To address this difficulty, I develop a novel and tractable method for constructing confidence regions when the possible DGPs can be non-identical. Especially, it brings back the convenient features in the case of i.i.d. and meanwhile guarantees the confidence regions to accommodate the non-identical DGPs with at least the same asymptotic level as for i.i.d. DGPs. 
	
	Formally, for any $p \in \Delta(S)$, let $p^{\infty}$ denote the i.i.d. distribution over $\Omega$ with its marginal equals $p$. Let $E_{N,\alpha} (p^{\infty})$ denote the region of acceptance for $p^{\infty}$ that satisfies the following condition:\footnote{More precisely, suppose they are constructed by the normal approximation.} 
	\begin{equation*}
		\liminf_{n \rightarrow \infty} p^{\infty} (E_{N,\alpha}(p^{\infty})) \geq 1-\alpha.
	\end{equation*}
	
	Recall that for any $P \in \Delta_{indep}(\Omega)$, $\bar{P}_{N} \in \Delta(S)$ denotes the average of sample marginals. Then let $E_{N, \alpha}((\bar{P}_{N})^{\infty})$ denote the region of acceptance that satisfies the previous condition under the i.i.d. distribution $(\bar{P}_{N})^{\infty}$. Then consider the following updating rule. 
	
	\begin{definition}
		The updated sets are given by the \textbf{robust i.i.d. statistical tests} with an asymptotic level $1-\alpha$ if for all $\omega_{N}$, 
		\begin{equation*}
			\mathcal{P}(\omega_{N}) = \{P \in \mathcal{P}: \omega_{N} \in E_{N,\alpha}((\bar{P}_{N})^{\infty})\}.
		\end{equation*}
	\end{definition}
	
	More explicitly, the robust i.i.d. statistical tests are given by the following two steps: 
	
	\begin{enumerate}
		\item For any sample data $\omega_{N}$, construct a confidence region as if the possible DGPs are all i.i.d.. 
		
		\item For every DGP in the initial set, retain it in the updated set if its average of sample marginals equals the marginal of an i.i.d. distribution in the previous confidence region. 
		
	\end{enumerate}
	
	Notice that the first step is simply the standard procedure of constructing confidence regions when the DGPs are all i.i.d.. The essential departure from the case of i.i.d. is the second step where it \textit{robustifies} the confidence region obtained under the additional i.i.d. assumption by also including the possible non-identical DGPs accordingly. Importantly, this robustification step is given by a simple comparison. Thus it does not add much computational difficulty. In other words, implementing the robust i.i.d. statistical tests is as tractable as conducting standard statistical inferences with i.i.d. DGPs. 
	
	The following theorem shows that under any possible non-identical DGP, the robust i.i.d. statistical tests can guarantee the updated sets to accommodate the truth with at least the same confidence level as for the i.i.d. DGPs. 
	
	\begin{theorem}\label{confident}
		The updated sets given by the robust i.i.d. statistical tests with an asymptotic level $1-\alpha$ accommodate the truth with the same asymptotic level. 
	\end{theorem}
	
	The proof of this theorem relies on a key observation: For any independent but non-identical distribution $P$, let $E_{N, \alpha}(P)$ be constructed by the probability contour of the corresponding multivariate Gaussian distribution such that the following holds (provided it satisfies the central limit theorem), 
	\begin{equation*}
		\liminf_{N \rightarrow \infty} P(E_{N,\alpha}(P)) \geq 1-\alpha.
	\end{equation*}
	Moreover, let $E_{N,\alpha}((\bar{P}_{N})^{\infty})$ denote the probability contour of the multivariate Gaussian distribution corresponding to the i.i.d. distribution $(\bar{P}_{N})^{\infty}$. The key observation is that for all $P \in \Delta_{indep}(\Omega)$, for all $N$ and  $\alpha$, 
	\begin{equation*}
		E_{N, \alpha}(P) \subseteq E_{N,\alpha}((\bar{P}_{N})^{\infty}).
	\end{equation*}
	Notice the inclusion further implies that
	\begin{equation*}
		P(E_{N, \alpha}(P))  \leq P(E_{N,\alpha}((\bar{P}_{N})^{\infty})).
	\end{equation*}
	It therefore leads to the conclusion that for testing the null hypothesis $H_{P}: P^{*} = P$, if the region of acceptance is given by $E_{N,\alpha}((\bar{P}_{N})^{\infty})$, then the probability of accepting $P$ when it is true will be at least greater than the probability when using $E_{N,\alpha}(P)$. The latter is furthermore asymptotically greater than $1-\alpha$ by construction. Therefore, the confidence regions constructed by the robust i.i.d. statistical tests can guarantee at least the same coverage probability for all non-identical DGPs. Moreover, this lower bound is also tight as i.i.d. DGPs can be present. 
	
	The key relation, $E_{N, \alpha}(P) \subseteq E_{N,\alpha}((\bar{P}_{N})^{\infty})$, is shown by deriving a result relating the average covariance matrices of the two distributions, $P$ and $(\bar{P}_{N})^{\infty}$. I show that for any $N$, subtracting the average covariance matrix of $P_{N}$ from the average covariance matrix of $(\bar{P}_{N})^{\infty}_{N}$ will always result in a positive semidefinite matrix. This result generalizes a well-known result for binomial distributions to the case of multinomial distributions. This well-known result is that the average variance of an i.i.d. binomial distribution is always weakly greater than the average variance of a non-identical binomial distribution whose average mean equals the mean of the i.i.d. distribution \citep{wang1993}. This generalization, to the best of my knowledge, is firstly discovered in the present paper. \\
	
	Some additional remarks: 
	
	\textbf{Remark 1.} Notice that the updated sets given by the robust i.i.d. statistical tests are relatively more conservative than those constructed directly. Specifically, it is more likely to accept $H_{P}$ when $P^{*} \neq P$. In other words, the probability of type \Rmnum{2} error for each individual test is more significant. On the one hand, such a drawback may not be an essential issue when the true DGP can only be partially identified, and in particular, when robustness is a major concern. On the other hand, the gain from using only the i.i.d. statistical tests is quite obvious. 
	
	\textbf{Remark 2.} It might also be a concern that finding the probability contours of the multivariate Gaussian distributions when $|S|$ is large can be difficult. An alternative option is to use Bonferroni's type of confidence region given by constructing \textit{confidence intervals} for the probability of every outcome $s_{j}$ with confidence level $1-\alpha/(d-1)$. Then the intersection of all such confidence intervals will be a confidence region with the confidence level $1-\alpha$. However, such a confidence region is much more conservative than the one constructed using the multivariate distribution. In this case, showing that Bonferroni's type of confidence region constructed by i.i.d. distributions guarantees at least the same coverage probability for non-identical distributions can be directly implied by invoking the corresponding result in \cite{wang1993} for binomial distributions. 
	
	\section{Applications}\label{application}
	The decision framework studied in this paper where a decision-maker observes sample data and makes decisions under model uncertainty is often used to model economic problems such as dynamic portfolio choice \citep{Epstein2007}, asset pricing \citep*{Epstein2008, Illeditsch2011}, and social learning \citep*{DEFILIPPIS2021105188, Chen2019}, among many others. 
	
	In those models, the existing papers often obtain conclusions by assuming that the DM applies either maximum likelihood updating or full Bayesian updating. For maximum likelihood updating, the present paper uncovers that it can lead to unreliable inferences and decisions. Thus the conclusions there may be misguided by this issue. Full Bayesian updating, on the other hand, might be too conservative. Especially, it often implies that learning may be ineffective under ambiguity. In this section, I show that this conclusion crucially hinges on the assumption of full Bayesian updating. 
	
	In the following, I will first study a specific but commonly studied model to show that learning can be, in fact, very effective under ambiguity if applying the average-then-update rule. Second, I will use a stylized example of Bernoulli models with ambiguous nuisance parameters to provide a concrete illustration of both updating rules proposed in this paper. Importantly, both models also highlight the tractability of the updating rules. 
	
	\subsection{Gaussian Signals with Ambiguous Variances}\label{Gaussian}
	
	Asymptotic results under full Bayesian updating is often very hard to solve and thus only known for some specific models. The section considers one of those models. Specifically, it is a generalization of the setting in \citet*{reshidi2020information}, which, in turn, generalizes the model of asset pricing under ambiguity studied in \cite{Epstein2008}. 
	
	A DM aims to learn the state of the world $\theta \in \Theta \equiv \R$ by observing  a countably infinite sequence of signals denoted by $\{x_{i}\}_{i=1}^{\infty}$. Each $x_{i}$ is a Gaussian random variable with mean $\theta$ and variance $\sigma_{i}^{2}$, and let $g_{i}(\theta, \sigma_{i})$ denote its probability density function. The sequence of random variables is mutually independent but every $\sigma_{i}$ is only known to belong to the interval $[\underline{\sigma}, \overline{\sigma}]$. In other words, the signals are given by a sequence of independent but possibly heterogeneous Gaussian random variables. 
	
	For each state $\theta$, let $\mathcal{P}_{\theta}$ denote the set of all possible data-generating processes over the sequence of signals, i.e., 
	\begin{equation*}
		\mathcal{P}_{\theta} = \left\{\prod_{i=1}^{\infty} g_{i}(\theta, \sigma_{i}) : \sigma_{i} \in [\underline{\sigma}, \overline{\sigma}] , \forall i \right\}.
	\end{equation*}
	The initial set of DGPs is thus given by 
	\begin{equation*}
		\mathcal{P} = \cup_{\theta \in \Theta} \mathcal{P}_{\theta}.
	\end{equation*}
	For every $N \in \N$, let $\hat{x}^{N} \equiv (\hat{x}_{1}, \hat{x}_{2}, \cdots, \hat{x}_{N})$ denote a sequence of signal realizations and let $\mathcal{P}(\hat{x}^{N} )$ denote the updated set of DGPs. Because the DM's goal in this model is to learn the state of the world, let $\Theta(\hat{x}^{N} )$ denote the set of states that are compatible with the updated set of DGPs: 
	\begin{equation*}
		\Theta(\hat{x}^{N} ) \equiv \{\theta \in \Theta: \mathcal{P}_{\theta} \cap \mathcal{P}(\hat{x}^{N} ) \neq \emptyset \}.
	\end{equation*}
	In other words, as long as there exists a DGP that is compatible with the state $\theta$ in the updated set, then the state $\theta$ will be considered possible. 
	
	First, notice that applying full Bayesian updating here is to simply retain all the possible DGPs. It thus implies $\mathcal{P}(\hat{x}^{N} ) \equiv \mathcal{P}$ and hence $\Theta(\hat{x}^{N} ) \equiv \Theta$ for all $\hat{x}^{N}$. In other words, directly applying full Bayesian updating is equivalent to completely ignoring the signals in this model. 
	
	In \citet*{reshidi2020information}, they apply full Bayesian updating differently by further assigning a prior distribution $\mu$ over the possible states of the world.\footnote{While they treat the prior distribution as an additional assumption, it can also be considered as a part of the updating rule.} Then applying full Bayesian updating is to apply Bayes' rule to update this prior distribution considering every possible DGP. Specifically, for each $\theta$ let $P_{\theta} \in \mathcal{P}_{\theta}$ denote a specific DGP, the posterior probability $\mu(\theta | \hat{x}^{N})$ is then given by 	
	\begin{equation*}
		\mu(\theta | \hat{x}^{N}) = \frac{\mu(\theta) P_{\theta} (\hat{x}^{N})}{\int_{\Theta}\mu(\theta') P_{\theta'} (\hat{x}^{N})d\mu(\theta')}.
	\end{equation*}	
	
	Full Bayesian updating thus leads to a set of posteriors over the states of the world. Their main result (Theorem 1) shows that under any state $\theta$ and possible DGP, the set of posteriors converges almost surely to a set of degenerate distributions over a non-singleton set of states denoted by $[\underline{m}(\theta), \overline{m}(\theta)]$. In other words, applying full Bayesian updating by assigning a prior distribution over the states still asymptotically almost surely leads to a fully ambiguous belief over a non-vanishing set of possible states. In other words, ambiguity does not vanish. 
	
	For applying the average-then-update rule, because the goal in this model is to learn the state of the world $\theta$, instead of considering the empirical distribution of the realized signals, it is sufficient to consider only their sample mean. In this case, the average-then-update rule can be applied in the following way: 
	\begin{definition}
		The updated sets of states are given by applying the average-then update rule if for some pre-specified $\epsilon > 0$ and for all $\hat{x}^{N}$, 
		\begin{equation*}
			\Theta(\hat{x}^{N} ) = \left\{\theta \in \Theta: \left|N^{-1}\sum\limits_{i=1}^{N}\hat{x}_{i} - \theta \right| < \epsilon \right\}.
		\end{equation*}
	\end{definition}
	
	By definition, the updated set of states retains a state if it is close enough to the sample mean. Notice that the ``average'' is given by considering the average of the mean of the marginal distributions. In this simple case, the mean of the marginal distributions are all equal to $\theta$. The following proposition shows that applying average-then-update here can guarantee the updated sets of states to contain the true state asymptotically almost surely under any possible DGP. 
	
	\begin{proposition}\label{gaussian}
		For any $\theta^{*} \in \Theta$ and any DGP $P^{*} \in \mathcal{P}_{\theta^{*}}$, the updated sets of states given by applying the average-then-update rule with any $\epsilon > 0$ contain the true state asymptotically almost surely, i.e., for any $\epsilon > 0$, 
		\begin{equation*}
			\lim\limits_{N \rightarrow \infty} P^{*}(\theta^{*} \in \Theta(\hat{x}^{N})) = 1.
		\end{equation*}
	\end{proposition}
	
	\begin{proof}[Proof of Proposition \ref{gaussian}]
		This proposition is proved by verifying Kolmogorov's strong law of large numbers holds, which then implies the sample mean converges almost surely to the average mean, which is $\theta^{*}$ for any possible DGP under the state $\theta^{*}$. See Appendix \ref{SLLN} for the exact definition and condition. 
	\end{proof}
	
	Because the conclusion holds for any $\epsilon > 0$, one can then let the $\epsilon$ to be arbitrarily small so that the updated set of states is arbitrarily precise.  As a result, even under a sequence of Gaussian signals with unknown and possibly heterogeneous variances, the true state can still be asymptotically almost surely learned from the realized signals. This conclusion stands in clear contrast to those obtained under full Bayesian updating. Therefore, learning can indeed be very effective even under ambiguity. 
	
	More importantly, the use of average-then-update in this model is justified in terms of improving decisions by the characterization results provided in this paper. A common motivation for applying full Bayesian updating is that the DM may concern about robustness and thus update all the possible DGPs. The main finding in this paper shows that one does not need to be that conservative and can still guarantee a robust improvement in decisions. 
	
	In an asset pricing model, consider that $\theta$ represents the dividend of an asset, and the Gaussian signal $x_{i}$ represents intangible information such as news reports. When there is only one signal realization, and the representative agent applies full Bayesian updating, \cite{Epstein2008} show that the equilibrium prices when the variance of $x_{i}$ is ambiguous are different from those when the variance is precisely known. They further conclude that poor information quality, i.e., ambiguity about variances of the signals, has an additional effect on the equilibrium prices for assets in a financial market.
	
	Proposition \ref{gaussian} implies that such an effect will eventually vanish when the number of signals goes to infinity and the representative agent applies the average-then update rule.\footnote{Instead, if the agent still applies full Bayesian updating, their result implies that such an effect will persist. See \cite{guidolin2013ambiguity} for a summary.} In this case, the agent will asymptotically learn the true dividend. This is the same as when the agent precisely knows the variances of the Gaussian signals. In other words, the prior ambiguity over the variances of the signals will be eventually ``swamped'' by a sufficiently large number of observations. Therefore the conclusion here implies that information quality per se does not have a persistent impact on the equilibrium prices of the assets. 
	
	It is true that all the Gaussian signals have the same mean is a crucial driving force for the asymptotic identification of the true state. Indeed, when there is also ambiguity about the mean of the signals, it is straightforward to show that, under the average-then-update rule, the DM can still asymptotically be ambiguous about a non-vanishing set of states. However, the essential takeaway here is that, average-then-update provides a tractable tool for studying the asymptotic learning under such ambiguity. The asymptotic result under full Bayesian updating in this case, to the best of my knowledge, is unknown in the literature. 
	
	Finally, the conclusion of Proposition \ref{gaussian} can also be obtained by applying the average-then-update rule directly using the empirical distribution. Notice that because the random variable $x_{i}$ takes value in an infinite set $\R$, the conclusion cannot be directly implied by Theorem \ref{asymptotic}. Instead, it can be shown by invoking a generalized Glivenko-Cantelli theorem for independent but non-identically distributed random variables derived by \cite{WELLNER1981309}. 
	
	\subsection{Bernoulli Models with Ambiguous Nuisance Parameters}\label{Bernoulli}
	To provide a more concrete illustration of both updating rules proposed in this paper, I consider another stylized model proposed in \cite{Walley1991-WALSRW}. A similar model is used to model dynamic portfolio choice under ambiguity in \cite{Epstein2007}. 
	
	Let $\Omega = \{1, 0\}^{\infty}$. It is convenient to consider a set of \textit{structural parameters} denoted by $\Theta = [0,1]$. Each structural parameter $\theta$ can be thought of as a rough estimate about the probability of getting the outcome $1$ from an experiment. However, the exact probability for each experiment will be also determined by a \textit{nuisance parameter} $\psi_{i} \in [0,1]$ that can be different across experiments. Specifically, for any structural parameter $\theta$, the probability of observing outcome $1$ in the $i$-th experiment is given by
	\begin{equation*}
		(1-\delta)\theta + \delta \psi_{i}
	\end{equation*}
	for some fixed $\delta \in [0,1]$. Hereafter, I use the probability of outcome $1$ to denote a probability distribution over $\{1,0\}$. 
	
	Because each $\psi_{i}$ is only known to be contained in the interval $[0,1]$, each structural parameter $\theta$ is thus corresponding to a set of possible DGPs given by the following
	\begin{equation*}
		\mathcal{P}_{\theta} = \{P \in \Delta_{indep}(\Omega): P_{i} \in [(1-\delta)\theta, (1-\delta)\theta + \delta], \forall i\}.
	\end{equation*}
	The initial set of DGPs is therefore
	\begin{equation*}
		\mathcal{P} = \cup_{\theta \in \Theta} \mathcal{P}_{\theta}.
	\end{equation*}
	
	The learning goal in this model is to predict the realizations of future experiments. For any DGP $P$, I use $P_{i = N+1}$ to denote the marginal distribution over the $(N+1)$-th experiment. The same notation is also used for sets of DGPs. Then the initial prediction about the $(N+1)$-th experiment is simply,
	\begin{equation*}
		\mathcal{P}_{i = N+1} = [0,1].
	\end{equation*}
	
	First, consider the DM's asymptotic prediction using the average-then-update rule. For simplicity, I will ignore the pre-specified $\epsilon$ by taking it to be arbitrarily small. Then the updated set will be given by the following: 
	\begin{equation*}
		\mathcal{P}(\omega_{N}) = \left\{P \in \mathcal{P}: N^{-1} \sum\limits_{i=1}^{N}P_{i} = \Phi(\omega_{N}) \right\}.
	\end{equation*}
	For any $\theta \in \Theta$ there exists $P \in \mathcal{P}_{\theta}$ that satisfies the above equation if and only if 
	\begin{equation*}
		\Phi(\omega_{N}) \in [(1-\delta)\theta, (1-\delta)\theta + \delta].
	\end{equation*}
	As a result, the updated set of structural parameters, defined by 
	\begin{equation*}
		\Theta(\omega_{N}) \equiv \{\theta \in \Theta: \mathcal{P}_{\theta} \cap \mathcal{P}(\omega_{N}) \neq \emptyset \}
	\end{equation*}
	is then given by 
	\begin{equation*}
		\Theta(\omega_{N})  = \left[ \max \left\{ \frac{\Phi(\omega_{N}) - \delta}{1-\delta} , 0 \right\} , \min \left\{ \frac{\Phi(\omega_{N}) }{1-\delta}, 1\right\} \right].
	\end{equation*}
	Notice the updated prediction is also completely determined by the updated set of structural parameters and given by 
	\begin{equation*}
		\mathcal{P}(\omega_{N})_{i = N+1} = \left[ \max \left\{ \Phi(\omega_{N}) - \delta , 0 \right\} , \min \left\{ \Phi(\omega_{N}) + \delta, 1\right\} \right].
	\end{equation*}
	Therefore, the asymptotic prediction for the future experiment under the average-then-update rule is simply the $\delta$ ``fattening'' of the observed empirical frequency.
	
	For a comparison, consider the asymptotic prediction under the maximum likelihood updating which can be shown to be 
	\begin{equation*}
		\mathcal{P}^{ML}(\omega_{N})_{i = N+1} = \left[ \max \left\{ \Phi(\omega_{N}) - (1-\Phi(\omega_{N}))\delta , 0 \right\} , \min \left\{ \Phi(\omega_{N}) + \Phi(\omega_{N})\delta , 1\right\} \right].
	\end{equation*}
	Notice it is always a subinterval of the prediction under average-then-update. As a result, the true marginal distribution governing the future experiment may not be covered by the asymptotic prediction given by maximum likelihood updating. For example, when
	\begin{equation*}
		\theta = \frac{\Phi(\omega_{N}) - \delta}{1-\delta} > 0
	\end{equation*}
	and $\psi_{i} = 1$ for all $i$ from $1$ to $N$. The empirical frequency also converges to $\Phi(\omega_{N})$. However, if $\psi_{N+1} = 0$, then the true marginal distribution over the $(N+1)$-th experiment is $(\Phi(\omega_{N}) - \delta )$ which is not covered by the prediction under the maximum likelihood updating. Therefore, the resulting decisions can be strictly worse than simply ignoring the data, also confirming the observation in the illustrative example. \\
	
	Next, consider the finite-sample prediction with an asymptotic level $1-\alpha$ given by applying the robust i.i.d. statistical tests. For any sample data $\omega_{N}$, the first step is to construct the confidence interval as if the underlying DGPs are all i.i.d. binomial distributions. Specifically, the corresponding confidence interval here is the \textit{Wilson Interval}\footnote{Different from the commonly used Wald Interval which uses the sample variance, Wilson Interval is constructed by directly inverting the statistical tests, thus using the null variance. Wilson Interval has considerably better asymptotic performance than the Wald Interval. See \citet*{Brown2001} for an extensive discussion.}. Let $z_{\alpha/2}$ denote the upper $100(\alpha/2) \%$ quantile of the standard normal distribution. Let $[\underline{W}(\omega_{N}), \overline{W}(\omega_{N})]$ denote the Wilson Interval which has the following closed-form expression: 
	\begin{align*}
		& \overline{W}(\omega_{N}) = \frac{N \Phi(\omega_{N}) + z_{\alpha/2}^{2}/2}{N + z_{\alpha/2}^{2}} + \frac{z_{\alpha/2} N^{1/2}}{N + z_{\alpha/2}^{2}} \left( \Phi(\omega_{N})(1- \Phi(\omega_{N})) + z_{\alpha/2}^{2}/(4N) \right)^{1/2}, \\
		& \underline{W}(\omega_{N}) = \frac{N \Phi(\omega_{N}) + z_{\alpha/2}^{2}/2}{N + z_{\alpha/2}^{2}} - \frac{z_{\alpha/2} N^{1/2}}{N + z_{\alpha/2}^{2}} \left( \Phi(\omega_{N})(1- \Phi(\omega_{N})) + z_{\alpha/2}^{2}/(4N) \right)^{1/2}.
	\end{align*}
	
	Given the i.i.d. confidence interval, the second step is to consider non-identical DGPs whose average of sample marginals falls into this confidence interval. It can be shown that a structural parameter $\theta$ is retained in the updated set if and only if 
	\begin{equation*}
		[(1-\delta)\theta, (1-\delta)\theta + \delta] \cap [\underline{W}(\omega_{N}), \overline{W}(\omega_{N})] \neq \emptyset.
	\end{equation*}
	Thus the updated set is then given by
	\begin{equation*}
		\Theta(\omega_{N})  = \left[ \max \left\{ \frac{\underline{W}(\omega_{N}) - \delta}{1-\delta} , 0 \right\} , \min \left\{ \frac{\overline{W}(\omega_{N})}{1-\delta}, 1\right\} \right].
	\end{equation*}
	Similarly, the updated prediction in this case is 
	\begin{equation*}
		\mathcal{P}(\omega_{N})_{i = N+1} = \left[ \max \left\{ \underline{W}(\omega_{N}) - \delta , 0 \right\} , \min \left\{\overline{W}(\omega_{N}) + \delta, 1\right\} \right].
	\end{equation*}
	Notice that it is also simply the $\delta$ ``fattening'' of the Wilson Interval. Both the asymptotic and finite-sample predictions under the proposed updating rules in this model have tractable expressions and also intuitive interpretations. 
	
	\section{Concluding Remarks}\label{conclusion}
	This paper emphasizes the role of statistical inference in making decisions facing a set of possible distributions. It borrows and combines methodologies from two strands of literature, statistical decision theory and dynamic decisions under ambiguity, and shows that sample data can robustly and objectively improve such decisions if and only if the inference from data accommodates the true DGP. 
	
	When there is limited knowledge about how sample data is drawn from a population of heterogeneous individuals, the decision-maker might contemplate a set of independent but possibly non-identical distributions. In this case, this paper uncovers that common inference methods, such as maximum likelihood and Bayesian updating often fail to accommodate the true DGP. 
	
	To address this problem, this paper develops two novel and tractable updating rules, average-then-update and robust i.i.d. statistical tests. They are shown to guarantee the updated sets to accommodate the true DGP and thus robustly improve decisions either asymptotically almost surely or in finite sample with a pre-specified confidence level. 
	
	Finally, this paper also studies a general decision framework that can be easily adapted for applications involving model uncertainty and ambiguity. The proposed updating rules prove to be tractable tools for studying those problems. More importantly, applying the proposed updating rules often leads to different conclusions from the existing models. Thus, it creates an avenue for new research questions and findings. This paper explores two examples along this direction. Formal development is left for future research. 
	
	\newpage
	\begin{appendices}
		\counterwithin{table}{section}
		
		\section{Additional Discussions}
		
		\subsection{The Illustrative Example}\label{examplediscussion}
		In this subsection, I will show that applying Bayesian updating in the illustrative example also asymptotically almost surely leads to rule out the true DGP, i.e., assigns a vanishingly small posterior probability to the true DGP. 
		
		For simplicity, suppose the DM's prior distribution over the possible DGPs are supported on the following set of DGPs: $\{0.6, 1\}^{\infty} \cup \{(1/3)^{\infty}\}$ where 
		\begin{equation*}
			\{0.6, 1\}^{\infty}  \equiv \{P \in \Delta_{indep}(\{1,0\}): P_{i} \in \{0.6, 1\}, \forall i\},
		\end{equation*}
		Then for each $N$, let $\{0.6, 1\}^{N} \cup \{(1/3)^{N}\}$ denote the set of possible marginals over the sample experiments given the initial set. Notice that such a set is finite with a cardinality of $(2^{N}+1)$ and hence also changes with the sample size. Thus, for each sample data $\omega_{N}$, in order to apply Bayesian updating, one needs to specify a prior distribution $\mu_{N}$ over the set of such marginals. Then applying Bayesian updating conditioning on the data will lead to a posterior over such marginals. 
		
		Consider when the prior distribution $\mu_{N}$ is a uniform distribution, i.e. $\mu_{N}(P_{N}) \equiv 1/(2^{N}+1)$ for all $P$. Then the posterior probability of each marginal over sample experiments will be given by
		\begin{equation*}
			\mu_{N}(P_{N}| \omega_{N}) = \frac{P_{N}(\omega_{N})\mu_{N}(P_{N})}{\sum\limits_{P'_{N}}  P'_{N}(\omega_{N})\mu_{N}(P'_{N})} = \frac{P_{N}(\omega_{N})}{\sum\limits_{P'_{N}}  P'_{N}(\omega_{N})}.
		\end{equation*}
		Notice it is proportional to the likelihood of observing the sample data. Therefore when the true DGP is $(1/3)^{\infty}$ and as $N$ goes to infinity, the posterior probability of the true marginal will be arbitrarily small compared to some other possible marginal in the set $\{0.6, 1\}^{N}$. As a result, the posterior probability of the true DGP under Bayesian updating will asymptotically almost surely goes to zero. 
		
		One might notice that while the posterior asymptoticly almost surely rules out the true DGP, the prior distribution also assigns a vanishingly small probability to the true DGP. This is not a problem if the DM applies the maxmin expected-utility criterion for making the data-free decision. However, if the DM chooses to apply the \textit{expected-utility} criterion using the prior distribution, then data-free and data-driven decisions are going to coincide in this example. 
		
		Consider the following example which further shows that under Bayesian updating, data can still asymptotically almost surely lead to strictly worse decisions even under the expected-utility criterion. \\
		
		\textbf{Example.} Let $\Omega = \{1,0\}^{\infty}$ and consider the following initial set of DGPs: 
		\begin{equation*}
			\{0.4, 0.5\}^{\infty}  \cup \{0.6, 1\}^{\infty}.
		\end{equation*}
		To facilitate discussion, let
		\begin{equation*}
			\mathcal{P}_{1} = \{0.4, 0.5\}^{\infty}  \text{ and } \mathcal{P}_{2}= \{0.6, 1\}^{\infty}.
		\end{equation*}
		
		Suppose initially the DM considers a uniform distribution over all the possible DGPs. Thus for each $N$, the prior distribution $\mu_{N}$ over the marginals over sample experiments is thus given by the following: 
		\begin{equation*}
			\mu_{N}(P_{N}) = 1/2^{N+1}.
		\end{equation*}
		Consider a decision problem $D^{N} = \{f,x\}$ where $f$ depends only on the $(N+1)$-th experiment and pays 1 and 0 in state 1 and 0 respectively. $x$ pays a constant payoff of 0.55. Notice that the expected payoff from $f$ under the prior distribution is 
		\begin{equation*}
			\frac{1}{2} \times  \left(\frac{1}{2} \times 0.4 + \frac{1}{2}\times 0.5 \right) + \frac{1}{2} \times \left(\frac{1}{2} \times 0.6 + \frac{1}{2}\times 1 \right) = 0.625 > 0.55.
		\end{equation*}
		The data-free decision according to the expected-utility criterion is thus to choose $f$ over $x$. 
		
		Suppose the true DGP is $(0.6)^{\infty}$, thus asymptotically the DM almost surely observes sample data with an empirical frequency of $1$ close to $0.6$. The observation here is that, the posterior distribution given such data will eventually concentrate on the marginals in $\mathcal{P}_{1}$. Notice that the average likelihood of observing any data under $\mathcal{P}_{1}$ is the same as the likelihood under the i.i.d. distribution $(0.45)^{\infty}$. The same also applies to $\mathcal{P}_{2}$ and it is the same as $(0.8)^{\infty}$. Because $0.45$ is closer to $0.6$ in terms of the Kullback-Leibler divergence, the well-known result by \cite{Berk1966} further implies the posterior distribution will eventually concentrate on  $(0.45)^{\infty}$, i.e., the marginals in $\mathcal{P}_{1}$. 
		
		According to such a posterior distribution, notice that the expected payoff of the act $f$ will be close to $0.45$, strictly lower than the payoff from $x$. Thus the data-driven decision will be choosing $x$, although it is strictly worse than the DM's objective payoff from her data-free decision of choosing $f$ (0.6). 
		
		\subsection{Minimax Regret Criterion}\label{regret}
		This subsection presents an example where the DM applies the minimax regret criterion. 
		
		Let $\Omega = \{1, 0\}$ and I use the probability of $1$ to denote any probability distribution over $\Omega$. Let the initial set be $\mathcal{P} = [0,1]$, i.e. the set of all possible probability distributions over $\Omega$. 
		
		Consider a basic decision problem $D_{b} = \{f, x\}$ where $f$ pays $1$ in state $1$ and pays $0$ in state $0$ and $x$ always pays $2/3$. According to the initial set, the maximum regret from taking these two actions are given by: 
		\begin{equation*}
			R(f) = \max\limits_{P \in \mathcal{P}} [P(1) (1-1) + P(0) (2/3 - 0)] = 2/3,
		\end{equation*}
		
		\begin{equation*}
			R(x) = \max\limits_{P \in \mathcal{P}} [P(1) (1-2/3) + P(0) (2/3 - 2/3)] = 1/3.
		\end{equation*}
		
		According to the minimax regret criterion, the data-free decision is thus to choose $x$. Next, suppose the DM's updated set is given by $\mathcal{P}' = [3/5,1]$. Then for the same basic decision problem, one has
		\begin{equation*}
			R'(f) = \max\limits_{P \in \mathcal{P}'} [P(1) (1-1) + P(0) (2/3 - 0)] = 4/15,
		\end{equation*}
		
		\begin{equation*}
			R'(x) = \max\limits_{P \in \mathcal{P}'} [P(1) (1-2/3) + P(0) (2/3 - 2/3)] = 1/3.
		\end{equation*}
		Therefore, the data-driven decision becomes to choose $f$. However, when $P^{*} \in [3/5, 2/3]$ which is contained in the updated set, objective payoff is higher from the data-free decision. 
		
		\newpage
		\section{Proofs of Results in Section \ref{characterization}}
		As this section studies the static problem of comparing two MEU decisions, to simplify notations, consider the following equivalent formulation of the problems in this section. 
		
		Let $\Omega$ be a state space and $\Delta(\Omega)$ be the set of all the countably additive probability measures over $\Omega$ endowed with the weak* topology. Let $\mathcal{P}$ be a closed and convex subset of $\Delta(\Omega)$ and $\mathcal{P}'$ a closed and convex subset of $\mathcal{P}$. Fix any $P^{*} \in \mathcal{P}$. 
		
		Let $f: \Omega \rightarrow [0,1]$ denote an act and $\mathcal{F}$ the set of all acts. A compact set $D \subseteq \mathcal{F}$ is a decision problem. Let $\mathcal{D}$ denote the set of all decision problems. Let $c(D)$ and $c'(D)$ denote the MEU decisions based on $\mathcal{P}$ and $\mathcal{P}'$. Let $D_{b}$ denote a basic decision problem and let $\mathcal{D}_{b}$ be defined analogously. 
		
		To see this is equivalent, one can let $S^{N} = \Omega$. The closed and convex hull of the marginals over $S^{N}$ for the initial and updated sets are given by $\mathcal{P}$ and $\mathcal{P}'$ respectively. Because $\mathcal{P}$ and $\mathcal{P}'$ are arbitrary, the additional independence requirement can always be satisfied. 
		
		\subsection{Proof of Theorem \ref{basic}}
		Under the simplified formulation, this theorem is equivalent to the following claim: 
		
		\begin{claim}
			$W(c'(D_{b}), P^{*}) \geq W(c(D_{b}), P^{*})$ for all $D_{b} \in \mathcal{D}_{b}$ if and only if $P^{*} \in \mathcal{P}'$. 	Moreover, the strict inequality holds for some $D_{b}$ if and only if $\mathcal{P}' \subsetneqq \mathcal{P}$. 
		\end{claim}
		
		To prove the claim, first consider the ``if'' direction.\\
		
		\textbf{IF.} Consider any basic decision problem $D_{b}  = \{f,x\}$. Notice that if $c(D_{b} ) = c'(D_{b} )$, then the conclusion holds trivially. Thus it suffices to consider $D_{b}$ where the two decisions differ. 
		
		Fix any such $D_{b}$, because $\mathcal{P}' \subseteq \mathcal{P}$, it has to be the case that $c(D_{b} ) = x$ and $c'(D_{b} ) = f$. Notice that when $f$ and $x$ are indifferent under $c(\cdot)$, because of the tie-breaking assumption (Assumption \ref{assumption4}), if $c(D_{b} ) = f$, then $c'(D_{b} ) = f$; if $c(D_{b} ) = x$, then $c'(D_{b} )$ can be either $x$ or $f$. 
		
		Given $P^{*} \in \mathcal{P}'$, for any act $f$ in such $D_{b}$ one has, 
		\begin{equation*}
			\int_{\Omega} f(\omega) dP^{*}(\omega) \geq \min\limits_{P \in\mathcal{P}'} \int_{\Omega} f(\omega)dP(\omega) > x,
		\end{equation*}
		where the second strict inequality is because $c'(D_{b} ) = f$ and $c(D_{b} ) = x$. (If indifferent, then $c'(D_{b} ) = x$.) 
		
		It thus implies 
		\begin{equation*}
			W(c'(D_{b} ), P^{*}) > x = W(c(D_{b} ), P^{*})
		\end{equation*}
		for all such $D_{b}  \in \mathcal{D}_{b}$. \\
		
		\textbf{ONLY IF.} Consider the contrapositive statement. If $P^{*} \notin \mathcal{P}'$, applying the strict separating hyperplane theorem implies the existence of an act $f$ and a constant act $x$ such that 
		\begin{equation*}
			\min\limits_{P \in \mathcal{P}'} \int_{\Omega} f(\omega) dP > x > \int_{\Omega} f(\omega) dP^{*} \geq \min\limits_{P \in \mathcal{P} } \int_{\Omega} f(\omega) dP(\omega).
		\end{equation*}
		
		Then for the basic decision problem $D_{b} = \{f,x\}$, one has $c(D_{b}) = x$ and $c'(D_{b}) = f$. Moreover, 
		\begin{equation*}
			W(c'(D_{b}), P^{*}) < x = W(c(D_{b}) ,P^{*}).
		\end{equation*}
		Thus, the conclusion. \\
		
		Finally, for the proof of the second statement of the claim, it can be routinely shown by applying the strict separating hyperplane theorem. $\square$
		
		\subsection{Proof of Theorem \ref{impossible}}\label{impossible_proof}
		This theorem is an implication of the following slightly more general theorem. 
		\begin{theorem}
			For any two DGPs $P^{*}$ and $\hat{P}^{*}$ with $P^{*N} \neq \hat{P}^{*N}$, objective dominance across all decision problems holds simultaneously under both DGPs if and only if $co(\mathcal{P}(\omega_{N})^{N}) = co(\mathcal{P}^{N})$. 
		\end{theorem}
		
		Under the simplified formulation, this theorem is equivalent to the following claim: 
		
		\begin{claim}
			Given $P^{*} \neq \hat{P}^{*}$, $W(c'(D), P^{*}) \geq W(c(D), P^{*})$ and $W(c'(D), \hat{P}^{*}) \geq W(c(D), \hat{P}^{*})$ for all $D \in \mathcal{D}$ if and only if $\mathcal{P} = \mathcal{P}'$. 
		\end{claim}
		
		The proof of this claim relies on the following lemma:
		
		\begin{lemma}\label{objectivedominancelemma}
			$W(c'(D), P^{*}) \geq W(c(D), P^{*})$ for all $D \in \mathcal{D}$ if and only if there exists $\alpha \in [0,1]$ such that $\mathcal{P}' = \alpha P^{*} + (1-\alpha) \mathcal{P}$. 
		\end{lemma}
		
		Given this lemma, notice that if $P^{*} \neq \hat{P}^{*}$, the only way for $\mathcal{P}'$ to satisfy the required relation for both of them is when $\alpha = 0$, i.e. $\mathcal{P}'  = \mathcal{P}$. Therefore, this proves the claim. $\square$
		
		\begin{proof}[Proof of Lemma \ref{objectivedominancelemma}]
			\textbf{IF.} The inequality is always true if $c(D) = c'(D)$. Thus fix any decision problem $D$ where the two decisions are different, let $f$ and $f'$ denote $c(D)$ and $c'(D)$ respectively. Given the tie-breaking assumption (Assumption \ref{assumption4}), this is possible only when $\alpha \neq 0$, thus hereafter $\alpha > 0$. 
			
			By definition, one has 
			\begin{equation}\label{objectivedominanceinequality1}
				\min\limits_{P\in \mathcal{P}} \int_{\Omega} f(\omega)dP(\omega) \geq \min\limits_{P\in \mathcal{P}} \int_{\Omega} f'(\omega)dP(\omega),
			\end{equation}
			and 
			\begin{equation}\label{objectivedominanceinequality2}
				\min\limits_{P\in \mathcal{P'}} \int_{\Omega} f'(\omega)dP(\omega) \geq \min\limits_{P\in \mathcal{P'}} \int_{\Omega} f(\omega)dP(\omega).
			\end{equation}
			
			Plugging the premise condition to the inequality \eqref{objectivedominanceinequality2} to get, 
			\begin{equation*}
				\alpha \int_{\Omega} f'(\omega) dP^{*}(\omega) + (1-\alpha) \min\limits_{P\in \mathcal{P}} \int_{\Omega} f'(\omega) dP(\omega) \geq \alpha \int_{\Omega} f(\omega) dP^{*}(\omega) + (1-\alpha) \min\limits_{P\in \mathcal{P}} \int_{\Omega} f(\omega) dP(\omega).
			\end{equation*}
			
			Rearranging terms to get, 
			\begin{equation*}
				\alpha \left[\int_{\Omega} f'(\omega) dP^{*}(\omega) - \int_{\Omega} f(\omega) dP^{*}(\omega)\right] \geq (1-\alpha) \left[\min\limits_{P\in \mathcal{P}} \int_{\Omega} f(\omega)dP(\omega) -\min\limits_{P\in \mathcal{P}} \int_{\Omega} f'(\omega)dP(\omega) \right] \geq 0,
			\end{equation*}
			where the last inequality follows from inequality \eqref{objectivedominanceinequality1}. As $\alpha > 0$, one has, 
			\begin{equation*}
				W(c'(D),P^{*}) - W(c(D),P^{*}) \geq 0.
			\end{equation*}
			\vspace{1mm}
			
			\textbf{ONLY IF.} Consider the contrapositive statement: If there does not exist any $\alpha \in [0,1]$ such that the equation holds, then there exists $D$ with $W(c'(D), P^{*}) < W(c(D), P^{*})$. 
			
			If $P^{*} \notin \mathcal{P}'$, then the conclusion can be implied by Theorem \ref{basic}. Consider the case where $P^{*} \in \mathcal{P}'$. By assumption, there exists $\alpha \in [0,1]$ such that 
			\begin{equation*}
				\mathcal{P}' \subsetneqq \alpha P^{*} + (1-\alpha) \mathcal{P},
			\end{equation*}
			and also some $P' \in \mathcal{P}' $ such that for all $\beta > \alpha$, 
			\begin{equation*}
				P' \notin \beta P^{*} + (1-\beta) \mathcal{P}.
			\end{equation*}
			
			In other words, $\mathcal{P}'$ is tangent to $\alpha P^{*} + (1-\alpha) \mathcal{P}$ at $P'$. Thus, there exists an act $f$ whose minimum expectation among both $\mathcal{P}'$ and $\alpha P^{*} + (1-\alpha) \mathcal{P}$ are all achieved at $P'$. It therefore implies that, 
			\begin{equation*}
				\min\limits_{P \in \mathcal{P}'} W(f, P)= \min\limits_{P \in \alpha P^{*} + (1-\alpha)\mathcal{P}} W(f,P) = \alpha W(f, P^{*}) + (1-\alpha) \min\limits_{P \in \mathcal{P}} W(f,P).
			\end{equation*}
			
			Next, because $\mathcal{P}' \subsetneqq \alpha P^{*} + (1-\alpha) \mathcal{P}$, by the strict separating hyperplane theorem, there also must exist an act $g$ such that for some $\beta > \alpha$, 
			\begin{equation*}
				\alpha W(g, P^{*}) + (1-\alpha) \min\limits_{P \in \mathcal{P}} W(g, P) < \min\limits_{P \in \mathcal{P}'} W(g, P)
				= \beta  W(g, P^{*}) + (1-\beta) \min\limits_{P \in \mathcal{P}} W(g, P)
			\end{equation*}
			where $\beta > \alpha$ follows by the fact that for such $g$, $W(g, P^{*}) > \min\limits_{P \in \mathcal{P}} W(g,P)$. 
			
			Normalize $f$ and $g$ by taking mixtures with constant acts to get for some $\epsilon > 0$ with  $2\epsilon < \beta-\alpha$, 
			\begin{align*}
				\min\limits_{P \in \mathcal{P}} W(f,P) = 1/2 > 1/2-\epsilon  = \min\limits_{P \in \mathcal{P}}  W(g, P) \\
				W(f, P^{*}) =1 > 1-\epsilon =  W(g, P^{*}). 
			\end{align*}
			Then for the decision problem $D = \{f, g\}$. It is the case that $c(D) = f$ and $c'(D) = g$. The second claim follows from
			\begin{align*}
				\min\limits_{P \in \mathcal{P}'} W(f, P) &  = \alpha W(f, P^{*}) + (1-\alpha) \min\limits_{P \in \mathcal{P}}W(f, P)  \\
				& = \alpha + (1-\alpha)/2\\
				&< \beta (1-\epsilon) + (1-\beta)(1/2-\epsilon) \\
				& = \beta W(g, P^{*}) + (1-\beta) \min\limits_{P \in \mathcal{P}}W(g, P_{2}) \\
				&= \min\limits_{P \in \mathcal{P}'} W(g, P) 
			\end{align*}
			where the strict inequality is equivalent to $2\epsilon < \beta - \alpha$. But $W(f, P^{*}) > W(g, P^{*})$, i.e. the conclusion. 
			
		\end{proof}
		
		\subsection{Proof of Theorem \ref{certaintyequivalent}}
		Similarly, this theorem is equivalent to the following claim: 
		\begin{claim}
			$W(c'(D), P^{*}) \geq \min\limits_{P \in \mathcal{P}} W(c(D), P)$ for all $D \in \mathcal{D}$ if and only of  $P^{*} \in \mathcal{P}'$. 
		\end{claim}
		First consider the ``if'' direction. \\
		
		\textbf{IF.} If $P^{*} \in \mathcal{P}' \subseteq \mathcal{P}$, then for any $D$, 
		\begin{align*}
			W(c'(D), P^{*}) &\geq \min\limits_{P \in \mathcal{P}'} W(c'(D), P) \\
			& \geq \min\limits_{P \in \mathcal{P}'} W(c(D), P)\\
			& \geq \min\limits_{P \in \mathcal{P}} W(c(D), P)
		\end{align*}
		where the first inequality follows from $P^{*} \in \mathcal{P}'$, the second inequality follows from the fact that $c'(D)$ is optimal with respect to $\mathcal{P}'$, the third inequality follows from $\mathcal{P}' \subseteq \mathcal{P}$. \\
		
		\textbf{ONLY IF.} Consider the contrapositive statement and notice that by Theorem \ref{basic}, it further implies the existence of a basic decision problem $D_{b}$ with 
		\begin{equation*}
			W(c'(D_{b}), P^{*}) < 	W(c(D_{b}), P^{*}) = \min\limits_{P \in \mathcal{P}} W(c(D_{b}), P),
		\end{equation*}
		where the equality follows from the fact that $c(D_{b}) = x$ in this case. Thus the conclusion.  $\square$

		\newpage
		\section{Proofs of Results in Section \ref{updating}}
		
		\subsection{Proof of Theorem \ref{asymptotic}}\label{SLLN}
		The main argument is to show the possible DGPs satisfy the condition of Kolmogorov's strong law of large numbers. For reference, I provide the exact definition here. 
		
		\textbf{Kolmogorov's SLLN} Let $\{X_{i}\}$ be a sequence of independent random variables. Define $Y_{N} = N^{-1}\sum_{i=1}^{N}X_{i}$ and $\bar{\mu}_{N} = N^{-1}\sum_{i=1}^{N} \mu_{i}$. If $E[X_{i}] = \mu_{i}$ and $var(X_{i}) = \sigma_{i}^{2}$, 
		\begin{equation*}
			\lim_{N \rightarrow \infty} \sum\limits_{i=1}^{N} \frac{\sigma_{i}^{2}}{i^{2}} < \infty
		\end{equation*}
		then 
		\begin{equation*}
			Y_{N} - \bar{\mu}_{N} \rightarrow 0 \quad a.s.
		\end{equation*}
		as $N \rightarrow \infty$. \\
		
		Fix any possible DGP $P \in \Delta_{indep}(\Omega)$. For any $s_{j} \in S$, let $X_{i} = I\{\omega_{i} = s_{j}\}$. Then $X_{i}$'s are independent random variables with $E[X_{i}] = P_{i}(s_{j})$ and $var(X_{i}) = P_{i}(s_{j})(1-P_{i}(s_{j}))$. To check the condition for the SLLN holds, 
		\begin{align*}
			\lim_{N \rightarrow \infty} \sum\limits_{i=1}^{N} \frac{\sigma_{i}^{2}}{i^{2}} & =  \lim_{N \rightarrow \infty} \sum\limits_{i=1}^{N} \frac{P_{i}(s_{j})(1-P_{i}(s_{j}))}{i^{2}} \\
			& \leq \lim_{N \rightarrow \infty} \sum\limits_{i=1}^{N} \frac{1}{i^{2}} =  \frac{\pi^{2}}{6}< \infty.
		\end{align*}
		Therefore, it implies that $\Phi(\omega_{N})(s_{j})- \bar{P}_{N}(s_{j}) \rightarrow 0$ almost surely for any $s_{j} \in S$. \\
		
		SLLN implies that for $P$-almost every $\omega$ and for any $\epsilon > 0$, there exists $\bar{N}(\omega, \epsilon, s_{j})$ such that for all $N \geq \bar{N}(\omega, \epsilon, s_{j})$, $|\Phi(\omega_{N})(s_{j})- \bar{P}_{N}(s_{j})| < \epsilon $.
		
		By definition, for any $\epsilon > 0$, the updated set given by the average-then-update rule can be written as 
		\begin{equation*}
			\mathcal{P}(\omega_{N}) = \{P \in \mathcal{P}: \cap_{s_{j} \in S} |\bar{P}_{N}(s_{j}) -  \Phi(\omega_{N})(s_{j})| < \epsilon \}.
		\end{equation*}
		
		Let $\bar{N}(\omega, \epsilon) = \max\limits \{\bar{N}(\omega,\epsilon,s_{j}\}$ (it is finite as $|S|$ is finite), and then SLLN implies that for all $N \geq \bar{N}(\omega, \epsilon)$, $P \in \mathcal{P}(\omega_{N})$. Thus, for any $\epsilon > 0$, $\mathcal{P}(\omega_{N})$ contains the truth, hence accommodates the truth asymptotically. $\square$

		\subsection{Proof of Theorem \ref{confident}} 
		The proof of this theorem is established by the following two lemmas which will be proved later. 
		
		\begin{lemma}\label{lem1}
			For any $P \in \mathcal{P}$, a central limit theorem holds so that the regions of acceptance can be constructed by the probability contours of the corresponding multivariate Gaussian distributions. 
		\end{lemma}
		
		\begin{lemma}\label{lem2}
			For any $P \in \mathcal{P}$ and any $N$, let $E_{N, \alpha} (P)$ and $E_{N, \alpha} ((\bar{P}_{N})^{\infty})$ be the two probability contours of the corresponding Gaussian distributions such that (this is guaranteed by Lemma \ref{lem1}), 
			\begin{align*}
				&\liminf_{n \rightarrow \infty} P(E_{N,\alpha}(P)) \geq 1-\alpha\\
				&\liminf_{n \rightarrow \infty} (\bar{P}_{N})^{\infty} (E_{N,\alpha}((\bar{P}_{N})^{\infty})) \geq 1-\alpha
			\end{align*}
			Then, 
			\begin{equation*}
				E_{N,\alpha}(P) \subseteq E_{N,\alpha}((\bar{P}_{N})^{\infty}).
			\end{equation*}
		\end{lemma}
		
		~~~
		
		\textit{Proof of the Theorem.} Lemma \ref{lem2} implies that for any $P \in \mathcal{P}$ and any $N$, 
		\begin{equation*}
			P( E_{N,\alpha}(P)) \leq P(E_{N,\alpha}((\bar{P}_{N})^{\infty})).
		\end{equation*}
		By construction, the following is true, 
		\begin{equation*}
			\liminf_{n \rightarrow \infty} P(E_{N,\alpha}(P)) \geq 1-\alpha.
		\end{equation*}
		Therefore, one can further derive,
		\begin{equation*}
			\liminf_{n \rightarrow \infty}  P(E_{N,\alpha}((\bar{P}_{N})^{\infty})) \geq \liminf_{n \rightarrow \infty} P(E_{N,\alpha}(P)) \geq 1-\alpha.
		\end{equation*}
		
		Finally, notice that 
		\begin{align*}
			\liminf_{n \rightarrow \infty}P(\{\omega_{N}: P_{N} \in co(\mathcal{P}(\omega_{n})_{N}) \}) & \geq \liminf_{n \rightarrow \infty} P(\{ \omega_{N}: P \in \mathcal{P}(\omega_{n}) \}) \\
			&= \liminf_{n \rightarrow \infty} P(E_{N,\alpha}((\bar{P}_{N})^{\infty})) \\
			&\geq 1-\alpha
		\end{align*}
		
		Thus the conclusion. $\square$\\
		
		Before proving the two lemmas, some additional notations are needed. 
		
		Recall that $S$ is finite. Let $S = \{s_{1}, \cdots, s_{d}\}$ and for $s = s_{1}, \cdots, s_{d-1}$, let $e_{s}$ denote the unit vector in the corresponding Euclidean space $\R^{d-1}$. 
		
		I use $\mathbf{X_{i}}: \Omega \rightarrow \R^{d-1}$ to denote a $(d-1)$-dimensional random vector such that  $\mathbf{X_{i}} = e_{\omega_{i}}$ if $\omega_{i} = s_{1},\cdots, s_{d-1}$ and $\mathbf{X_{i}} = 0$ if $\omega_{i} = s_{d}$. For each $P$, let $\mathbf{P_{i}}$ denote the $(d-1)$-dimensional vector such that the corresponding coordinates are given by $P_{i}(s)$ for $s = s_{1},\cdots,s_{d-1}$. 
		
		If $P$ is the underlying probability distribution, then the mean vector of $\mathbf{X_{i}}$ is given by $\boldsymbol\mu_{i} = \mathbf{P_{i}}$. Moreover, the covariance matrix $\boldsymbol\Sigma_{\mathbf{i}}$ of the random vector $\mathbf{X_{i}}$ is given by 
		
		\begin{equation*}
			[\boldsymbol\Sigma_{\mathbf{i}}]_{kl} = 
			\begin{cases}
				P_{i}(s_{k})(1-P_{i}(s_{k})) \quad &\text{ if } k=l\\
				-P_{i}(s_{k})P_{i}(s_{l}) \quad &\text{ if } k \neq l
			\end{cases}
		\end{equation*}
		
		Notice that by the full-support assumption, the covariance matrix is always positive definite. Let $\mathbf{\bar{X}_{N}} \equiv N^{-1}\sum_{i=1}^{N} \mathbf{X_{i}}$ be the random vector given by the average of the sequence of random vectors. Let $\boldsymbol{\bar{\mu}_{N}} \equiv N^{-1}\sum_{i=1}^{N} \boldsymbol{\mu_{\mathbf{i}}}$ and $\boldsymbol{\bar{\Sigma}_{N}} \equiv N^{-1} \sum_{i=1}^{N}\boldsymbol\Sigma_{\mathbf{i}}$. \\
		
		\textit{Proof of Lemma \ref{lem1}.} The proof is to show a multivariate central limit theorem (CLT) for independent but non-identically distributed random vectors holds for all possible DGP $P \in \mathcal{P}$. The proof is rather standard. I provide it here only for completeness. The reader should feel free to skip this proof if can be convinced it is true.\\
		
		Specifically, for any $P \in \mathcal{P}$, want to show  $\sqrt{n}(\mathbf{\bar{X}_{N{}}} - \boldsymbol{\bar{\mu}_{N}}) \xrightarrow{D} N\left(0, \boldsymbol{\bar{\Sigma}_{N{}}} \right)$, the standard argument is given by combining the Liapounov's CLT for Triangular arrays with the Cramer-Wold device. The following definitions are taken from \cite{WHITE1984107}. \\
		
		\textbf{Liapounov's CLT for Triangular Arrays.} 
		Let $\{ Z_{Ni} \}$ be a sequence of independent random scalars with $\mu_{Ni} \equiv E[Z_{Ni}]$, $\sigma_{Ni}^{2} = var(Z_{Ni})$, and $E|Z_{Ni} - \mu_{Ni}|^{2+\delta} < \Delta < \infty$ for some $\delta > 0$ and all $N$ and $i$. Define $\bar{Z}_{N} \equiv N^{-1}\sum_{i=1}^{N}Z_{Ni}$, $\bar{\mu}_{N} \equiv N^{-1}\sum_{i=1}^{N}\mu_{Ni}$ and $\bar{\sigma}_{N}^{2} \equiv var(\sqrt{N} \bar{Z}_{N}) = N^{-1}\sum_{i=1}^{N} \sigma_{Ni}^{2}$. If $\bar{\sigma}_{N}^{2} > \delta' > 0 $ for all $N$ sufficiently large, then 
		\begin{equation*}
			\sqrt{N}(\bar{Z}_{N} - \bar{\mu}_{N})/\bar{\sigma}_{N} \xrightarrow{D} N(0,1).
		\end{equation*}
		
		\textbf{Cramer-Wold device.} Let $\{ \mathbf{b}_{N} \}$ be a sequence of random $k \times 1$  vectors and suppose that for any real $k \times 1$ vector $\boldsymbol\lambda$ such that $\boldsymbol\lambda'\boldsymbol\lambda = 1$, $\boldsymbol\lambda' \mathbf{b}_{N} \xrightarrow{D}\boldsymbol\lambda' \mathbf{Z}$ where $\mathbf{Z}$ is a $k\times 1$ vector with joint distribution function $F$. Then the limiting distribution of $ \mathbf{b}_{N}$ exists and equals $F$. \\
		
		For any $\boldsymbol\lambda$, consider the sequence of random variables given by $Y_{i} =  \boldsymbol\lambda'\mathbf{X_{i}}$. Then one has $E[Y_{i}] =\boldsymbol\lambda'\boldsymbol{\mu_{\mathbf{i}}}$ and $var(Y_{i}) = \boldsymbol\lambda' \boldsymbol\Sigma_{\mathbf{i}}\boldsymbol\lambda$. 
		
		Then for $\bar{Y}_{N} = N^{-1}\sum_{i=1}^{N} Y_{i}$, one has $E[\bar{Y}_{N}] = \boldsymbol\lambda' \boldsymbol{\bar{\mu}_{N}}$ and
		
		\begin{equation*}
			var(\sqrt{N}\bar{Y}_{N}) =  \boldsymbol\lambda' \boldsymbol{\bar{\Sigma}_{N}}\boldsymbol\lambda = \boldsymbol\lambda'(\boldsymbol{\bar{\Sigma}_{N}}^{1/2})'(\boldsymbol{\bar{\Sigma}_{N}}^{1/2})\boldsymbol\lambda.
		\end{equation*}
		
		The existence of $\boldsymbol{\bar{\Sigma}_{N}}^{1/2}$ is guaranteed by the fact that $\boldsymbol{\bar{\Sigma}_{N}}$ is positive definite. Moreover, let it be the positive square root of $\boldsymbol{\bar{\Sigma}_{N}}$ so that it is unique, symmetric and positive definite (Theorem 7.2.6 in \citet{horn2012matrix}). 
		
		Let $Z_{Ni} =\boldsymbol\lambda'\boldsymbol{\bar{\Sigma}_{N}}^{-1/2}\mathbf{X_{i}}$ . Then one has $E[\bar{Z}_{N}] = \boldsymbol\lambda'\boldsymbol{\bar{\Sigma}_{N}}^{-1/2}\boldsymbol{\bar{\mu}}$ and 
		
		\begin{align*}
			var(\sqrt{N}\bar{Z}_{N}) & =  \boldsymbol\lambda' (\boldsymbol{\bar{\Sigma}_{N}}^{-1/2})' var(\sqrt{N}\mathbf{\bar{X}}_{N})\boldsymbol{\bar{\Sigma}_{N}}^{-1/2}\boldsymbol\lambda  \\
			& =  \boldsymbol\lambda' (\boldsymbol{\bar{\Sigma}_{N}}^{-1/2})' (\boldsymbol{\bar{\Sigma}_{N}}^{1/2})'(\boldsymbol{\bar{\Sigma}_{N}}^{1/2})(\boldsymbol{\bar{\Sigma}_{N}}^{-1/2})\boldsymbol\lambda\\
			&= 1.
		\end{align*}
		
		If the Liapounov's CLT holds, then  $\sqrt{N}(\bar{Z}_{N} - E[\bar{Z}_{N}])\xrightarrow{D} N(0,1)$. Applying the Cramer-Wold device will imply the multivariate CLT. 
		
		Thus, it remains to show that the Liapounov condition holds for any $P \in \mathcal{P}$ and for any $\boldsymbol\lambda$, which is, for some $\delta > 0$, there exists $\Delta < \infty$ such that for any $N$ and $i$,
		\begin{align*}
			E|Z_{Ni}|^{2 + \delta} < \Delta.
		\end{align*}
		
		By Minkowski's inequality, 
		\begin{equation*}
			E[|Z_{Ni}|^{2 + \delta}] \leq \sum_{j}\lambda_{j} E|(\boldsymbol{\bar{\Sigma}_{N}}^{-1/2}\mathbf{X_{i}})_{j}|^{2+\delta}.
		\end{equation*}
		
		Notice that $\mathbf{X}_{i}$ is either $0$ or $e_{k}$. The RHS is strictly less than the following
		\begin{equation*}
			(\max_{k,l} [\boldsymbol{\bar{\Sigma}_{N}}^{-1/2}]_{kl})^{2+\delta}.
		\end{equation*}
		
		Since $[\boldsymbol{\bar{\Sigma}_{N}}^{-1/2}]$ is symmetric and positive semidefinite, the largest entry must be on the diagonal. The trace of the matrix is equal to the sums of the eigenvalues. Both the diagonal elements and eigenvalues are nonnegative due to positive definiteness (Corollay 7.1.5 in \cite{horn2012matrix}). Thus, it suffices to show the eigenvalues of $\boldsymbol{\bar{\Sigma}_{N}}^{-1/2}$ are bounded from above. 
		
		Moreover, as each eigenvalue of $\boldsymbol{\bar{\Sigma}_{N}}^{-1/2}$ is the square root of the eigenvalue of $\boldsymbol{\bar{\Sigma}_{n}}^{-1}$. Thus, one only needs to show the eigenvalues of $\boldsymbol{\bar{\Sigma}_{N}}^{-1}$ are bounded from above, which is equivalent to showing the eigenvalues of $\boldsymbol{\bar{\Sigma}_{N}}$ are bounded away from zero (Theorem 4.1.10 in \cite{horn2012matrix}). 
		
		For every $\boldsymbol\Sigma_{i}$, the smallest eigenvalue is bounded below by $\min_{j}P_{i}(s_{j})$ \citep{watson1996}. Furthermore, as each $\boldsymbol\Sigma_{i}$ is also symmetric and positive definite, the minimum eigenvalue of $\sum_{i=1}^{n} \boldsymbol\Sigma_{i}$ ($n$ times the minimum eigenvalue of $\boldsymbol{\bar{\Sigma}_{N}}$) is greater than the sum of the minimum eigenvalues of every $\boldsymbol\Sigma_{i}$ (Corollary 4.3.15 in \cite{horn2012matrix}). 
		
		Therefore, it suffices to have $P_{i}(s_{j})$ be bounded away from zero for every $i$ and $j$. By assumption, $\mathcal{P}$ is compact thus $\min_{j,i} P_{i}(s_{j})$ exists. By full-support assumption, the minimum is always positive. $\square$\\ 
		
		\textit{Proof of Lemma \ref{lem2}.} 
		For a $(d-1)$-dimensional random vector $\mathbf{X} \sim \mathbf{N}_{d-1}(\boldsymbol{\mu}, \boldsymbol\Sigma)$, the ellipsoidal region given by the set of all vectors $\mathbf{x}$ satisfying the following has a probability of $1-\alpha$: 
		\begin{equation*}
			(\mathbf{x} - \boldsymbol\mu)^{T}\boldsymbol\Sigma^{-1} (\mathbf{x} - \boldsymbol\mu) \leq \chi_{d-1}^{2}(\alpha),
		\end{equation*}
		where $\chi_{d-1}^{2}(\alpha)$ is the upper $100\alpha\%$ quantile for the chi-square distribution with $d-1$ degrees of freedom. By definition, the set of all such vectors is the probability contour of the multivariate Gaussian distribution $\mathbf{N}_{d-1}(\boldsymbol{\mu}, \boldsymbol\Sigma)$ with probability $1-\alpha$. 
		
		For any $P \in \Delta_{indep}(\Omega)$, let $E_{N,\alpha}(P)$ denote the probability contour of the multivariate Gaussian distribution $\mathbf{N}_{d-1}(\boldsymbol{\mu}_{N}, \bar{\boldsymbol\Sigma}_{N})$ with probability $1-\alpha$. 
		
		Let $\boldsymbol{\hat{\Sigma}_{N}}$ denote the covariance matrix of the i.i.d. distribution $(\bar{P}_{N})^{\infty}$. Then let $E_{N,\alpha}((\bar{P}_{N})^{\infty})$ denote the corresponding probability contour for the multivariate Gaussian distribution $\mathbf{N}_{d-1}(\boldsymbol{\mu}_{N}, \hat{\boldsymbol\Sigma}_{N})$. Notice that for the two Gaussian distributions, their mean vectors are the same but their covariance matrices are different. 
		
		To show that $E_{N,\alpha}(P) \subseteq E_{N,\alpha}((\bar{P}_{N})^{\infty})$ for any $N$ and $\alpha \in [0,1]$, it suffices to show that for any $N$, any constant $c$, and for all $\mathbf{x}$, 
		
		\begin{equation*}
			(\mathbf{x} - \boldsymbol\mu)^{T}\boldsymbol{\bar{\Sigma}_{N}}^{-1} (\mathbf{x} - \boldsymbol\mu)\leq c \Rightarrow \boldsymbol (\mathbf{x} - \boldsymbol\mu)^{T}\boldsymbol{\hat{\Sigma}_{N}} ^{-1} \boldsymbol (\mathbf{x} - \boldsymbol\mu) \leq c,
		\end{equation*}
		which is further equivalent to
		\begin{equation*}
			(\mathbf{x} - \boldsymbol\mu)^{T}\boldsymbol{\bar{\Sigma}_{N}}^{-1} (\mathbf{x} - \boldsymbol\mu)-\boldsymbol (\mathbf{x} - \boldsymbol\mu)^{T}\boldsymbol{\hat{\Sigma}_{N}}^{-1} \boldsymbol (\mathbf{x} - \boldsymbol\mu) \geq 0.
		\end{equation*}
		
		In other words, for any $\mathbf{x}$, 
		\begin{equation*}
			\mathbf{x}^{T} (\boldsymbol{\bar{\Sigma}_{N}}^{-1}-\boldsymbol{\hat{\Sigma}_{N}} ^{-1})  \mathbf{x}\geq 0.
		\end{equation*}
		
		That is, the lemma is true if $(\boldsymbol{\bar{\Sigma}_{N}}^{-1}-\boldsymbol{\hat{\Sigma}_{N}}^{-1})$ is a positive semidefinite matrix. It is further equivalent to, according to Corollary 7.7.4 in \cite{horn2012matrix}, $(\boldsymbol{\hat{\Sigma}_{N}} - \boldsymbol{\bar{\Sigma}_{N}})$ being positive semidefinite.
		
		Next show that it is indeed the case. The two covariance matrices are given by
		\begin{equation*}
			[\boldsymbol{\hat{\Sigma}}_{N}]_{kl} = 
			\begin{cases}
				\bar{P}_{N}(s_{k})(1-\bar{P}_{N}(s_{k})) \quad &\text{ if } k=l\\
				-\bar{P}_{N}(s_{k})\bar{P}_{N}(s_{l}) \quad &\text{ if } k \neq l
			\end{cases}
		\end{equation*}
		
		\begin{equation*}
			[\boldsymbol{\bar{\Sigma}}_{N}]_{kl} = 
			\begin{cases}
				N^{-1}\sum\limits_{i=1}^{N} P_{i}(s_{k})(1-P_{i}(s_{k})) \quad &\text{ if } k=l\\
				-N^{-1}\sum\limits_{i=1}^{N} P_{i}(s_{k})P_{i}(s_{l}) \quad &\text{ if } k \neq l
			\end{cases}
		\end{equation*}
		
		By algebra, one can show, 
		\begin{equation*}
			N[\boldsymbol{\hat{\Sigma}}_{N} - \boldsymbol{\bar{\Sigma}}_{N}]_{kl} = 
			\begin{cases}
				\sum\limits_{i=1}^{N} (P_{i}(s_{k}) - \bar{P}_{N}(s_{k}))^{2}\quad &\text{ if } k=l\\
				\sum\limits_{i=1}^{N} (P_{i}(s_{k}) - \bar{P}_{N}(s_{k}))(P_{i}(s_{l}) - \bar{P}_{N}(s_{l})) \quad &\text{ if } k \neq l
			\end{cases}
		\end{equation*}
		
		Notice that $N[\boldsymbol{\hat{\Sigma}}_{N} - \boldsymbol{\bar{\Sigma}}_{N}]_{kl}$ is a Gram Matrix $G_{kl} = <v_{k}, v_{l}>$ where the set of vectors are given by: 
		\begin{equation*}
			[v_{k}]_{i} = P_{i}(s_{k}) - \bar{P}_{N}(s_{k}).
		\end{equation*}
		A Gram Matrix is always positive semidefinite (Theorem 7.2.10 in \cite{horn2012matrix}), therefore $\boldsymbol{\hat{\Sigma}}_{N} - \boldsymbol{\bar{\Sigma}}_{N}$ is positive semidefinite as desired. $\square$
		
	\end{appendices}
	
	\newpage
	\bibliographystyle{econ-econometrica}
	\bibliography{/Users/xiaoyucheng/Dropbox/Research/References/mylibrary.bib}
	
\end{document}